\documentclass[a4paper, 11pt]{article}
\usepackage[utf8]{inputenc}

\usepackage[margin=1in]{geometry}

\usepackage{todonotes}
\usepackage{amsmath}
\usepackage{amssymb}
\usepackage{amsfonts}
\usepackage{amsthm}
\usepackage{authblk}
\usepackage{xspace}
\usepackage{paralist}
\usepackage{multirow}
\usepackage{graphicx}
\usepackage{tabularx}
\usepackage{mathtools}
\usepackage[numbers, sort&compress]{natbib}
\usepackage[pagebackref=True]{hyperref}
\usepackage{cleveref}

\usepackage{tikz}
\usetikzlibrary{arrows, calc, shapes, patterns}
\usetikzlibrary{decorations.pathreplacing}

\newtheorem{claim}{Claim}
\newtheorem{theorem}{Theorem}
\newtheorem{lemma}{Lemma}
\newtheorem{observation}{Observation}
\newtheorem{corollary}{Corollary}
\newtheorem{definition}{Definition}

\crefname{claim}{Claim}{Claims}
\crefname{corollary}{Corollary}{Corollaries}
\crefname{definition}{Definition}{Definitions}

\newcommand{\DTW}{\textsc{DTW-Mean}\xspace}
\newcommand{\fCCS}[1][]{\ensuremath{{f^{ \text{\tiny\sf CS}}_{#1}}}}
\newcommand{\fMSCS}[1][$f$]{\textsc{#1-MSCS}\xspace}
\newcommand{\fMSCSlong}[1][$f$]{\textsc{Multiple String Circular Shift (with Cost~#1)}\xspace}

\newcommand{\RMCClong}{\textsc{Regular Multicolored Clique}\xspace}
\newcommand{\RMCC}{\textsc{RMCC}\xspace}

\newcommand{\Cgap}{\ensuremath{C_\mathrm{gap}}}
\newcommand{\Ccode}{\ensuremath{C_\mathrm{code}}}
\newcommand{\Cback}{\ensuremath{C_\mathrm{back}}}

\newcommand{\N}{\mathbb{N}}
\newcommand{\Q}{\mathbb{Q}}

\newcommand{\Fcost}{\mathcal{F}}
\DeclareMathOperator{\cost}{cost}
\DeclareMathOperator{\dtw}{dtw}
\DeclareMathOperator{\poly}{poly}
\newcommand{\longZero}{\textbf{\textit{0}}}

\newcommand{\problemdef}[3]{
	\begin{center}
  \begin{minipage}{0.95\textwidth}
    \noindent
    \textsc{#1}

			\vspace{2pt}
			\setlength{\tabcolsep}{3pt}
			\begin{tabularx}{\textwidth}{@{}lX@{}}
					\textbf{Input:} 		& #2 \\
					\textbf{Question:} 	& #3
				\end{tabularx}
  \end{minipage}
	\end{center}
      }

\title{Tight Hardness Results for Consensus Problems on Circular Strings and Time Series}

\author[1]{Laurent Bulteau}
\author[2]{Vincent Froese}
\author[2]{Rolf Niedermeier}

\affil[1]{\small Université Paris-Est, LIGM (UMR 8049), CNRS, ENPC, ESIEE Paris, UPEM, F-77454, Marne-la-Vallée, France\\
laurent.bulteau@u-pem.fr}
\affil[2]{\small
  Algorithmics and Computational Complexity, Faculty~IV, TU Berlin, Berlin, Germany\\
\{vincent.froese, rolf.niedermeier\}@tu-berlin.de}

\date{}

\begin{document}

\maketitle

\begin{abstract}
  Consensus problems for strings and sequences appear in numerous application contexts, ranging from bioinformatics over data mining to machine learning.
Closing some gaps in the literature, we show that several fundamental problems in this context are NP- and W[1]-hard, and that the known (partially brute-force) algorithms are close to optimality assuming the Exponential Time Hypothesis.
Among our main contributions is to settle the complexity status of computing a mean in dynamic time warping spaces which, as pointed out by Brill~et~al.~[DMKD~2019], suffered from many unproven or false assumptions in the literature.
We prove this problem to be NP-hard and additionally show that a recent dynamic programming algorithm is essentially optimal.
In this context, we study a broad family of circular string alignment problems. This family also serves as a key for our hardness reductions, and it is of independent (practical) interest in molecular biology.
In particular, we show tight hardness and running time lower bounds for \textsc{Circular Consensus String}; notably, the corresponding non-circular version is easily linear-time solvable.

\medskip

\noindent\textbf{Keywords:} Consensus Problems; Time Series Analysis; 
Circular String Alignment;
Fine-Grained Complexity and Reductions; 
Lower Bounds; Parameterized Complexity;
Exponential Time Hypothesis.
\end{abstract}

\section{Introduction}\label{sec:intro}

Consensus problems appear in many contexts of stringology and time series 
analysis, including applications in bioinformatics, data mining, machine learning, and 
speech recognition. Roughly speaking, given a set of input sequences,
the goal is to find a consensus sequence that minimizes the ``distance''
(according to some specified distance measure) to the input sequences.
Classic problems in this context are the NP-hard \textsc{Closest String} \cite{FL97,LMW02a,LMW02b,GNR03}
(where the goal is to find a ``closest string'' that minimizes 
the maximum Hamming distance to a 
set of equal-length strings) or the more general \textsc{Closest Substring}~\cite{FGN06,Marx08}.
Notably, the variant of \textsc{Closest String} where one minimizes the sum of Hamming distances 
instead of the maximum distance is easily solvable in linear time.

In this work, we settle the computational complexity of prominent consensus  
problems on circular strings and time series.
Despite their great importance in many applications, and a correspondingly rich 
set of heuristic solution strategies used in practice, 
to date, it has been unknown whether 
these problems are polynomial-time solvable or NP-hard. 
We prove their hardness, including also ``tight'' parameterized and 
fine-grained complexity results, thus justifying the massive use of 
heuristic solution strategies in real-world applications.

On the route to determining the complexity of exact mean computation in dynamic time warping spaces, a fundamental consensus problem in the context of time series analysis~\cite{PKG11}\footnote{As of April~2019, the work by Petitjean et al.~\cite{PKG11}, who developed one of the most prominent heuristics for this problem, has already been cited around 340 times since 2011 according to Google Scholar.}, we first study a fairly general alignment problem\footnote{Particularly from the viewpoint of applications in bioinformatics, consensus string problems can also be interpreted as 
alignment problems \cite{LNPPS13}.} 
for circular strings called \fMSCSlong. Based on its analysis, we 
will also derive our results for two further, more specific problems.
Given a set of input strings over a fixed alphabet~$\Sigma$ and a
local cost function~$f\colon \Sigma^* \rightarrow \Q$,
the goal in \fMSCSlong{} (abbreviated by \fMSCS) is to find a cyclic shift of 
each input string such that the shifted strings
``align well'' in terms of the sum of local costs.\footnote{We cast all problems in this work as decision 
problems for easier complexity-theoretic treatment. Our hardness results correspondingly 
hold for the associated optimization problems.}

\problemdef{\fMSCS}
{A list of~$k$ strings $s_1,\ldots, s_k\in\Sigma^n$ of length~$n$ and~$c\in\Q$.}
{Is there a multiple circular shift $\Delta=(\delta_1,\ldots,\delta_k)\in \N^k$ with $\cost_f(\Delta):=\sum_{i=1}^n f\big((s_1^{\leftarrow\delta_1}[i], \ldots, s_k^{\leftarrow\delta_k}[i])\big) \le c$?}

See \Cref{fig:exampleFMSCS} for an example.
We separately study the special case 
\textsc{Circular Consensus String} for a binary alphabet, where the cost 
function~$f\colon \{0,1\}^*\rightarrow\N$ 
is defined as $f((x_1,\ldots, x_k)):= \min \{\sum_{i=1}^kx_i, k-\sum_{i=1}^kx_i\}$. This corresponds to minimizing the sum of Hamming 
distances (not the maximum Hamming distance as in \textsc{Closest String}).
As we will show, allowing circular shifts makes the problem much harder to solve.

\begin{figure}[t]
\centering
	\begin{tikzpicture}[xscale=.78, 
	dtw/.style={rectangle, rounded corners=1mm, inner sep=2pt, fill=white, draw=black, line width=1pt, minimum height=1.2em},
	lineb/.style={line width=2pt, red}, lineg/.style={line width=2pt, green}, 	
	]
	\foreach \y/\d in {1/0,2/2,3/1} {
		\coordinate (s\y) at (1,-\y);
		\coordinate (sv\y) at (9,-\y);
		\node at (s\y) {$s_{\y}$};
		\node at (sv\y) {$s^{\leftarrow\d}_{\y}$};
	}

\node  at (9,-3.6) {cost:};
	\foreach \l [count=\i] in { $\frac 23$, 0,0,0, $\frac 23$} {
		\node at (9+\i, -3.6) {\footnotesize \l};	   
	}

	\draw[lineb] ($(s1)+(1,0 )$) to[out=-90,in=90, looseness=.6] ($(s2)+(3,0)$) to[out=-90,in=90, looseness=.6]($(s3)+(2,0)$);
	\draw[lineg] ($(s1)+(2,0 )$) to[out=-90,in=90, looseness=.6] ($(s2)+(4,0)$) to[out=-90,in=90, looseness=.6]($(s3)+(3,0)$);
	
	\draw[lineg] ($(s1)+(3,0 )$) to[out=-90,in=90, looseness=.6] ($(s2)+(5,0)$) to[out=-90,in=90, looseness=.6]($(s3)+(4,0)$);
	
	\draw[lineg] ($(s1)+(4,0 )$) to[out=-90,in=160, looseness=.6] ($(s2)+(5.5,0.33)$);
	\draw[lineg] ($(s2)+(0.5, 0.33 )$) to[out=-20,in=90, looseness=.6] ($(s2)+(1,0)$) to[out=-90,in=45, looseness=.6]($(s3)+(.5,.5)$);
	\draw[lineg]  ($(s3)+(5.5,.5)$) to[out=-134,in=90, looseness=.6]($(s3)+(5,0)$);

	\draw[lineb] ($(s1)+(5,0 )$) to[out=-90,in=160, looseness=.6] ($(s2)+(5.5,0.66)$);
	\draw[lineb] ($(s2)+(0.5, 0.66 )$) to[out=-20,in=90, looseness=.6] ($(s2)+(2,0)$) to[out=-90,in=90, looseness=.6]($(s3)+(1,0)$);
	
	\foreach \s [count=\j] in {b,g,g,g,b} {
		\draw[line\s] ($(sv1)+(\j,0 )$) -- ($(sv3)+(\j,0)$);
		
	}
	
	\foreach \r [count=\y] in {{1,0,0,1,1}, {1,1,0,0,0}, {0,1,0,0,1}}{
		\foreach \v [count=\x] in \r{
			\node[dtw, draw=gray] at ($(s\y)+(\x, 0)$) {\v}; 
		} 
	} 
	\foreach \r [count=\y] in {{1,0,0,1,1}, {0,0,0,1,1}, {1,0,0,1,0}}{
		\foreach \v [count=\x] in \r{
			\node[dtw, draw=gray] at ($(sv\y)+(\x, 0)$) {\v}; 
			
		} 
	} 
	
	\end{tikzpicture}
	\caption{An instance of \fMSCS[$\sigma$] with three binary input strings, and an optimal multiple circular shift $\Delta=(0,2,1)$, using the sum of squared distances from the mean ($\sigma$) as a cost function. Columns of $\Delta$ are indicated with dark (red) or light (green) lines, depending on their cost. For example, column 1, with values $(1,0,1)$ has mean $\frac 23$, and cost $\left(\frac 13\right)^2+\left(\frac 23\right)^2+\left(\frac 13\right)^2=\frac 23$. The overall cost is~$\frac 43$.}
        \label{fig:exampleFMSCS}
\end{figure}
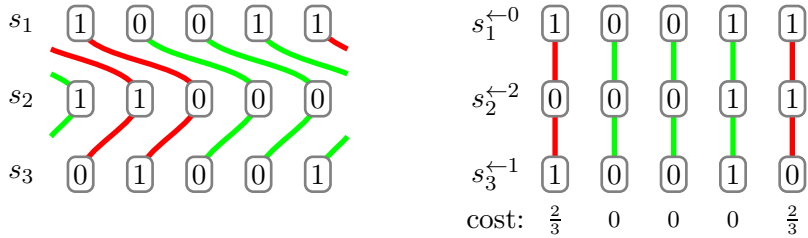

Multiple circular string (sequence) alignment problems have 
been considered in different variations
in bioinformatics, where circular strings naturally arise in several applications (for example, in multiple alignment of genomes, which often have a circular molecular structure)~\cite{MVC02,FPF09,WS14,BIKPRV15,GIMPPRV16,AP17}.
Depending on the application at hand, different cost functions are used.
For example, non-trivial algorithms for computing a consensus string of three and four circular strings with respect to the Hamming distance have been developed~\cite{LNPPS13}.
However, most of the algorithmic work so far is heuristic in nature or only
considers specific special cases. A thorough analysis of the computational complexity for these problems in general so far has been missing.

After having dealt with circular string alignment problems in a quite general fashion,
we then study a fundamental (consensus) problem in time series analysis.
\emph{Dynamic time warping} (see \Cref{sec:prelims} for details) defines a distance between two time series which is used in many applications in time series analysis~\cite{KR05,PKG11,SDG16,MAKD18} (notably, dynamic time warping has also been considered in the context of circular sequences~\cite{Arica05,PM12}).
An important problem here is to compute an average of a given sample of time series under the dynamic time warping distance.

\problemdef{\DTW}
{A list of~$k$ univariate rational time series $x_1,\ldots, x_k$ and~$c\in\Q$.}
{Is there a univariate rational time series~$z$ such that $\Fcost(z) = %\frac{1}{k}
\sum_{i=1}^{k}\left(\dtw(z,x_i)\right)^2 \le c$?}

\begin{figure}
  \centering	
  \begin{tikzpicture}[xscale=.78, 
    dtw/.style={rectangle, rounded corners=1mm, inner sep=2pt, fill=white, draw=gray, line width=1pt, minimum height=1.2em},
    dtwlineA/.style={line width=2pt, draw=black!85},
    dtwareaA/.style={dtwlineA, fill=black!25},	
    dtwlineB/.style={dtwlineA, draw=black!30},
    dtwareaB/.style={dtwlineB, fill=black!15}	
    ]
    \foreach \y in {1,2,3} {
      \coordinate (s\y) at (1,-\y);
      \coordinate (sv\y) at (9,-\y);
      \node at (s\y) {$x_{\y}$};
      \node at (sv\y) {$x_{\y}$};
    }
    \node (z) at (9,-4) {$z$};
    \foreach \l [count=\i] in { $\frac 14$, 1,10,0, $\frac 43$} {
      \coordinate (z\i) at ($(z)+(\i, 0)$);   
    }
    
    \node  at (9,-4.6) {cost:};
    
    \foreach \l [count=\i] in { $\frac 34$, 2,0,0, $\frac {32}3$} {
      \node at ($(z\i) +(0,-.6)$) {\footnotesize \l};	   
    }

    \draw[dtwareaA] (z1)to[out=90,in=-90, looseness=.6]  ($(sv3)+(1,0 )$) to[out=90,in=-90, looseness=.6] ($(sv2)+(1,0)$) to[out=90,in=-90, looseness=.6]($(sv1)+(0.9,0)$) to[out=-90,in=90, looseness=.6] ($(sv2)+(1,0)$)  to[out=-90,in=90, looseness=.6] ($(sv3)+(2,0)$)  to[out=-90,in=90, looseness=.6] (z1)  ;
    
    \draw[dtwareaB] (z4)to[out=90,in=-90, looseness=.6]  ($(sv3)+(4.94,0 )$) to[out=90,in=-90, looseness=.6] ($(sv2)+(4,0)$) to[out=90,in=-90, looseness=.6]($(sv1)+(3,0)$)--($(sv1)+(4,0)$) to[out=-90,in=90, looseness=.6] ($(sv2)+(4,0)$)  to[out=-90,in=90, looseness=.6] ($(sv3)+(4.94,0)$)  to[out=-90,in=90, looseness=.6] (z4)  ;
    
    \foreach \a/\b/\c/\z/\s in {1.07/2/3/2/B,2/3/4/3/A,5/5/5.14/5/A}{
      \draw[dtwline\s] (z\z)to[out=90,in=-90, looseness=.6]  ($(sv3)+(\c,0 )$) to[out=90,in=-90, looseness=.6] ($(sv2)+(\b,0)$) to[out=90,in=-90, looseness=.6]($(sv1)+(\a,0)$);
    }

    \foreach \r [count=\y] in {{1,10,0,0,4}, {0,2,10,0,0}, {0,0,0,10,0}}{
      \foreach \v [count=\x] in \r{
        \node[dtw] at ($(s\y)+(\x, 0)$) {\v}; 
        
        \node[dtw] at ($(sv\y)+(\x, 0)$) {\v}; 			
      } 
    } 
    \draw [fill=white, line width=1.5pt, opacity=0.9] ($(z1)+(-0.25,-0.3)$) rectangle ($(z5)+(0.25,0.3)$);
    \foreach \l [count=\i] in { $\frac 14$, 1,10,0, $\frac 43$} {
      \node[] at (z\i) {\l};	   
    }
  \end{tikzpicture}
  \caption{A \DTW instance with three input sequences and an optimal length-5 mean~($z$). Alignments between the mean and input sequences can progress at different speeds. This is formalized using \emph{warping paths} (see \Cref{sec:prelims}) represented by polygons (or lines in degenerate cases) with alternating shades. Every pair of aligned elements belongs to the same polygon.
  The cost of each mean element is the sum of squared differences over all aligned input elements, e.g.\ the cost of the first element is $(1-\frac 14)^2 + 3\cdot (0-\frac 14)^2 = \frac 34$.}
      \label{fig:exampleDTW}
\end{figure}

Here, $\dtw$ denotes the dynamic time warping distance (see \Cref{sec:prelims} for details).
Intuitively, dynamic time warping allows for non-linear alignments between two series.
\Cref{fig:exampleDTW} depicts an example.
The dtw-distance of two length-$n$ time series can be computed via standard dynamic programming in~$O(n^2)$ time.
A slight improvement to $O(n^2\frac{\log\log\log n}{\log\log n})$ time is known~\cite{GS17}.
For two binary time series, there exists an~$O(n^{1.87})$-time algorithm~\cite{ABW15}.
In general, however, a strongly subquadratic-time algorithm (that is, $O(n^{2-\varepsilon})$ time for some~$\varepsilon>0$) does not exist unless the Strong Exponential Time Hypothesis fails~\cite{ABW15,BK15}.

Regarding the computational complexity of \DTW, 
although more or less implicitly assumed in many 
publications presenting heuristic solution strategies\footnote{For instance, Petitjean et al.~\cite{PFWNCK16} write ``Computational biologists have long known that averaging under time warping is a very
complex problem, because it directly maps onto a multiple sequence alignment: the ``Holy Grail''
of computational biology.'' Unfortunately, the term ``directly maps'' has not been formally defined and only sketchy explanations are given.}, 
NP-hardness still has been open (see~Brill et al.~\cite[Section~3]{BFFJNS19} for a discussion 
on some misconceptions and wrong statements in the literature).
It is known to be solvable in~$O(n^{2k+1}2^kk)$ time, where~$n$ is the maximum length of any input series~\cite{BFFJNS19}. Moreover, Brill et al.~\cite{BFFJNS19} presented 
a polynomial-time algorithm for the special case of binary time series.
In practice, numerous heuristics are used~\cite{PKG11,PFWNCK16,CB17,SJ18}.
Note that \DTW is often described as closely related to multiple sequence alignment problems~\cite{PG12,ASW15,PG17}.
However, we are not aware of any formal proof regarding this connection.
By giving a reduction from \fMSCSlong to \DTW, we show that \DTW is actually connected to multiple \emph{circular} sequence alignment problems.
To the best of our knowledge, this is the first formally proven result regarding this connection.

\paragraph*{Our Results.}
Using plausible complexity-theoretic assumptions, we provide a fine-grained 
picture of the exact computational complexity (including parameterized complexity) 
of the problems introduced above.
We present two main results.

First, we show that, for a large class of natural cost functions~$f$, \fMSCS on binary sequences is NP-hard, W[1]-hard with respect to the number~$k$ of inputs, and not solvable in~$\rho(k)\cdot n^{o(k)}$ time for any computable function~$\rho$ (unless the Exponential Time Hypothesis fails).
Note that \fMSCS is easily solvable in~$\rho(k)\cdot n^{O(k)}$ time (for computable functions~$f$) since there are at most~$n^{k-1}$ cyclic shifts to try out (without loss of generality, the first string is not shifted).
Our running time lower bound thus implies that brute-force is essentially optimal (up to constant factors in the exponent).
Based on this, we can also prove the same hardness for the \textsc{Circular Consensus String} problem.
In fact, the general ideas of our reduction might also be used to develop hardness reductions for other circular string alignment problems.

As our second main contribution, we obtain the same list of hardness results as above 
for \DTW on binary time series. We achieve this by a polynomial-time reduction from a special case of \fMSCS.
Our reduction implies that, unless the Exponential Time Hypothesis fails, the known 
$O(n^{2k+1}2^kk)$-time algorithm~\cite{BFFJNS19} essentially can be improved only up to constants in the exponent.
Note that recently Buchin et al.~\cite{BDS19} achieved the same hardness result for the problem of averaging time series under generalized $(p,q)$-DTW. Their reduction, however, does not yield binary input time series.

\paragraph*{Organization.} 
In Section~\ref{sec:prelims} we fix notation and introduce basic concepts, 
also including the formal definition of dynamic time warping and the corresponding 
concept of warping paths.
In Section~\ref{sec:mscs}, we identify a circular string problem (of independent interest in molecular biology) which forms the basis for the results in \Cref{sec:dtw}.
More specifically, we prove the hardness results for 
\fMSCSlong. The key ingredient here is a specially geared reduction from the \textsc{Regular 
Multicolored Clique} problem. Moreover, we introduce the concept of 
polynomially bounded grouping functions~$f$ (for which our results hold).
In Section~\ref{sec:ccs}, providing a reduction from \fMSCSlong, we show 
analogous hardness results for \textsc{Circular Consensus String}. 
Notably, the cost function corresponding to \textsc{Circular Consensus String}
is not a polynomially bounded grouping function, making the direct 
application of the result for \fMSCSlong impossible.
In Section~\ref{sec:dtw} we prove analogous complexity results for \DTW{},
again devising a reduction from \fMSCSlong.
In Section~\ref{sec:concl}, we conclude with some open 
questions and directions for future research.

\section{Preliminaries}\label{sec:prelims}
In this section, we briefly introduce our notation and formal definitions.

\paragraph*{Circular Shifts.}
We denote the $i$-th element of a string~$s$ by~$s[i]$, and its length by $|s|$.
For a string $s\in\Sigma^n$ and~$\delta\in\N$, we define the \emph{circular (left) shift by~$\delta$} as the string 
\begin{align*}
 s^{\leftarrow\delta}:=
&s[\delta+1]\ldots s[n]s[1] \ldots s[\delta]
\quad (\text{that is, }  s^{\leftarrow\delta}[i] = s[(i+\delta-1 \bmod n)+1]),
 \end{align*}
that is, we circularly shift the string $\delta$ times to the left.
Let $s_1,\ldots,s_k$ be strings of length~$n$. A \emph{multiple circular (left) shift} of~$s_1,\ldots,s_k$ is defined by a $k$-tuple $\Delta=(\delta_1,\ldots,\delta_k)\in\{0,\ldots,n-1\}^k$ and yields the strings $s_1^{\leftarrow\delta_1},\ldots,s_k^{\leftarrow\delta_k}$.
We define \emph{column~$i\in\{1,\ldots,n\}$} of a multiple circular shift~$\Delta$ as the $k$-tuple $(s_1^{\leftarrow\delta_1}[i], \ldots, s_k^{\leftarrow\delta_k}[i])$.
By \emph{row $j\in\{1,\ldots,k\}$} of column~$i$ we denote the element~$s_j^{\leftarrow\delta_j}[i]$.

\paragraph*{Cost Functions.}

A \emph{local cost function} is a function $f\colon \Sigma^*\rightarrow\Q$ assigning a cost to any tuple of values.
Given such a function, the \emph{overall cost} of a circular shift~$\Delta$ for 
$k$~ length~$n$ strings is defined as
$$\cost_f(\Delta):=\sum_{i=1}^n f\big((s_1^{\leftarrow\delta_1}[i], \ldots, s_k^{\leftarrow\delta_k}[i])\big),$$
that is, we sum up the local costs of all columns of~$\Delta$.

For example, a well-known local cost is the sum of squared distances from the arithmetic mean (i.e.\ $k$ times the variance, here called~$\sigma$), that is,
\begin{align*}
 \sigma((x_1,\ldots,x_k))
   &=\sum_{i=1}^k \left(x_i- \frac{1}{k}\sum_{j=1}^kx_j\right)^2.
   %&= \frac{1}{k}\sum_{j=1}^kx_j^2 - \left(\frac{1}{k}\sum_{j=1}^kx_j\right)^2
   \end{align*}
Using a well-known formula for the variance, we get the following useful formula for computing~$\sigma$: 
$$\sigma((x_1,\ldots,x_k))= \bigg(\sum_{j=1}^kx_j^2 \bigg) - \frac{1}{k}\bigg(\sum_{j=1}^kx_j \bigg)^2$$

\paragraph*{Dynamic Time Warping.}
A time series is a sequence~$x=(x_1,\ldots,x_n)\in\Q^n$.
The dynamic time warping distance between two time series is based on the concept of a warping path.
\begin{definition}
  A \emph{warping path} of order~$m\times n$ is a sequence~$p=(p_1,\ldots,p_L)$, $L\in\N$,
  of index pairs $p_\ell=(i_\ell,j_\ell)\in \{1,\ldots,m\}\times\{1,\ldots,n\}$, $1\le \ell \le L$, such that
  \begin{compactenum}[(i)]
    \item $p_1=(1,1)$,
    \item $p_L=(m,n)$, and
    \item $(i_{\ell+1}-i_\ell, j_{\ell+1}-j_\ell)\in \{(1,0),(0,1),(1,1)\}$ for each~$1\le \ell \le L-1$.
  \end{compactenum}
\end{definition}
\noindent See \Cref{fig:exampleDTW} in \Cref{sec:intro} for an example.

The set of all warping paths of order~$m\times n$ is denoted by~$\mathcal{P}_{m,n}$.
A warping path~$p\in\mathcal{P}_{m,n}$ defines an \emph{alignment} between two time series~$x=(x[1],\ldots,x[m])$ and~$y=(y[1],\ldots,y[n])$ in the following way:
Every pair~$(i,j)\in p$ \emph{aligns} element~$x_i$ with~$y_j$. Note that 
every element from~$x$ can be aligned with multiple elements from~$y$, and vice versa.
The \emph{dtw-distance} between~$x$ and~$y$ is defined as
$$\dtw(x,y) := \min_{p\in\mathcal{P}_{m,n}}\sqrt{\sum_{(i,j)\in p}(x[i]-y[j])^2}.$$

For \DTW, the cost of a mean $z$ for $k$ input time series $x_1,\ldots,x_k$ is given by 
\begin{align*}
 \Fcost(z):=\sum_{j=1}^{k}\left(\dtw(z,x_j)\right)^2
 =\sum_{j=1}^k \min_{p_j\in\mathcal{P}_{|x_j|,|z|}}  \sum_{(u,v)\in p_j}(x_j[u]-z[v])^2.
\end{align*}
Note that for \DTW, a normalized cost $F(z):=\frac 1k\Fcost(z)$ is often considered: this does not affect the computational complexity of the problem, so for simplification purposes we only consider the non-normalized cost $\Fcost(z)$.

\paragraph*{Parameterized Complexity.}
We assume familiarity with the basic concepts from classical 
and parameterized complexity theory.

An instance of a parameterized problem is a pair~$(I,k)$ consisting of the classical problem instance~$I$ and a natural number~$k$ (the \emph{parameter}).
A parameterized problem is contained in the class~XP if there is an algorithm solving an instance~$(I,k)$ in polynomial time if~$k$ is a constant, that is, in time~$O(|I|^{f(k)})$ for some computable function~$f$ only depending on~$k$ (here~$|I|$ is the size of~$I$). A parameterized problem is \emph{fixed-parameter tractable} (contained in the class~FPT) if it is solvable in time~$f(k)\cdot|I|^{O(1)}$ for some computable function~$f$ depending solely on~$k$.
The class W[1] contains all problems which are parameterized reducible to \textsc{Clique} parameterized by the clique size.
A parameterized reduction from a problem~$Q$ to a problem~$P$ is an algorithm mapping an instance~$(I,k)$ of~$Q$ in time~$f(k)\cdot |I|^{O(1)}$ to an equivalent instance~$(I',k')$ of~$P$ such that~$k' \le g(k)$ (for some functions~$f$ and~$g$).
It holds FPT $\subseteq$ W[1] $\subseteq$ XP.

A parameterized problem that is W[1]-hard with respect to a parameter (such as \textsc{Clique} with parameter clique size) is presumably not in FPT.

\paragraph*{Exponential Time Hypothesis.}
Impagliazzo and Paturi~\cite{IP01} formulated the \emph{Exponential Time Hypothesis} (ETH) which asserts that there exists a constant~$c > 0$ such that \textsc{3-SAT} cannot be solved in~$O(2^{cn})$ time, where~$n$ is the number~$n$ of variables in the input formula. It is a stronger assumption than common complexity assumptions such as P$\neq$NP or FPT$\neq$W[1]. The \emph{Strong Exponential Time Hypothesis} states that \textsc{SAT} cannot be solved faster than~$O(2^n)$.

Several conditional running time lower bounds have since been shown based on the ETH, for example, \textsc{Clique} cannot be solved in~$\rho(k)\cdot n^{o(k)}$ time for any computable function~$\rho$ unless the ETH fails~\cite{CCFHJKX05}.

\section{Hardness of \fMSCS on Binary Strings}\label{sec:mscs}
In this section, we consider only binary strings from~$\{0,1\}^*$. We prove hardness for a family of local cost functions that satisfy certain properties. The functions we consider have the common property that they only depend on the number of~$0$'s and~$1$'s in a column, and that they aim at grouping similar values together.

\begin{definition}\label{def:canonical}
  A function $f\colon\{0,1\}^*\rightarrow\Q$ is called \emph{order-independent} if, for each $k\in\N$, there exists a function~$f_k\colon\{0,\ldots,k\}\rightarrow\Q$ such that $f((x_1,\ldots, x_k)) = f_k\big(\sum_{j=1}^kx_j\big)$ holds for all~$(x_1,\ldots,x_k)\in\{0,1\}^k$.

  For an order-independent function~$f$, we define the function $f'_k:\{1,\ldots, k\}\rightarrow \Q$ as $$f'_k(x)=\frac{f_k(x)-f_k(0)}{x}.$$ 
    
  \noindent An order-independent function~$f$ is \emph{grouping} if $f'_k(k) < \min_{1\le x < k}f'_k(x)$ and $f_k'(2)<f_k'(1)$ holds for every~$k\in\N$. 
\end{definition}

For an order-independent function, $f'_k$ can be seen as the cost per 1-value  (a column with $x$ 1's and $k-x$ 0's has cost $f_k(x)=f_k(0)+xf_k'(x)$). It can also be seen as a discrete version of the derivative for $f_k$, so that if $f_k$ is concave then $f_k'$ is decreasing.
The intuition behind a grouping function is that the cost per $1$-value is minimal in columns containing only $1$'s, and that having two $1$'s in a column has less cost than having two columns with a single~1. In particular, any concave function is grouping. Finally, if $f$ is grouping, then the cost function with value $f_k(x)+ax+b$ is also grouping for any values $a$ and $b$.

The following definitions are required to ensure that our subsequent reduction (\Cref{thm:MSCS_reduction}) is computable in polynomial time.
\begin{definition}
  Let $f$ be an order-independent function. The \emph{gap} of~$f_k$ is
  $$\varepsilon_k :=\min\{f'_k(x)-f'_k(y) \mid 1\le x,y\le k, f'_k(x) > f'_k(y)\}.$$
  The \emph{range} of~$f_k$ is $\mu_k:=\max_{1\le x \le k}|f'_k(x)|$. 

An order-independent function $f$ is \emph{polynomially bounded} if it is polynomial-time computable and if, for every~$k\in\N$, $\mu_k$ and $\varepsilon_k^{-1}$ are upper-bounded by a polynomial in~$k$.
\end{definition}

For binary strings, the function $\sigma$ (see \Cref{sec:prelims}) is a polynomially bounded grouping function. Indeed, it is order-independent since $\sigma((x_1,\ldots x_k)) = w - \frac{w^2}{k} =\frac {w(k-w)}{k}$, where $w=\sum_{j=1}^k x_j = \sum_{j=1}^k(x_j)^2$ since~$x_j\in\{0,1\}$ for $j=1,\ldots,k$.
Thus, $\sigma_k(w) = \frac {w(k-w)}k$ and
we have $\sigma_k(0) =0$, and $\sigma'_k(w)=\frac {k-w}{k}$, so $\sigma'_k$ is strictly decreasing, which is sufficient for $\sigma$ to be grouping. Finally, it is polynomially bounded, with gap $\varepsilon_k=\frac{1}{k}$ and range $\mu_k= \frac {k-1}k \leq 1$.

We prove our hardness results with a polynomial-time reduction from a special version of the \textsc{Clique} problem.

\problemdef{\RMCClong~(RMCC)}
{A $d$-regular undirected graph~$G=(V,E)$ where the vertices are colored with $k$~colors such that each color class contains the same number of vertices.}
{Does~$G$ have a size-$k$ complete subgraph (containing $\binom{k}{2}$ edges, 
called a $k$-clique) with exactly one vertex from each color?}

\noindent\RMCC is known to be NP-hard, W[1]-hard with respect to~$k$, and
not solvable in $\rho(k)\cdot|V|^{o(k)}$ time for any computable function~$\rho$ unless the ETH fails~\cite{Cyg15}. %{MS12}.  

The following lemma states the existence of a polynomial-time reduction from \RMCC to \fMSCS which implies hardness of \fMSCS for polynomially bounded grouping functions.

\begin{lemma}\label{thm:MSCS_reduction}
  Let~$f$ be a polynomially bounded grouping function.
  Then there is a polynomial-time reduction that, given an \RMCC instance~$G=(V,E)$ with~$k$ colors,
  outputs binary strings $s_0,\ldots,s_{k}$ of equal length and~$c\in\Q$ such that the following holds:
  \begin{compactitem}
  \item If~$G$ contains a properly colored $k$-clique, then there exists a multiple circular shift~$\Delta$ of~$s_0,\ldots,s_{k}$ with~$\cost_f(\Delta) = c.$

  \item If~$G$ does not contain a properly colored~$k$-clique, then every multiple circular shift~$\Delta$ of $s_0,\ldots,s_{k}$ has $\cost_f(\Delta)\ge c+ \varepsilon_{k+1}$.
    \end{compactitem}
\end{lemma}

To prove \Cref{thm:MSCS_reduction}, we first describe the reduction and then prove several claims
about the structure and the costs of multiple circular shifts in the resulting \fMSCS instance.

\paragraph{Reduction.}

  Consider an instance of \RMCC, that is, a graph  $G=(V,E)$ with a partition of $V$ into $k$ subsets $V_1,\ldots, V_k$ of size~$n:=\frac{|V|}{k}$ each, such that each vertex has degree $d$. 
  Let $V_j=\{v_{j,1}, \ldots, v_{j,n}\}$, $m=|E|$, and $E=\{e_1,\ldots,e_m\}$.
  We assume that~$k\ge 3$ since the instance is trivially solvable otherwise.

  We build an \fMSCS instance with $k+1$ binary strings, hence we consider the local cost function~$f_{k+1}$. For simplicity, we write $f'$, gap~$\varepsilon$, and range~$\mu$ for $f'_{k+1}$,  $\varepsilon_{k+1}$, and $\mu_{k+1}$.
  
  For each $j\in\{1,\ldots,k\}$, let $p_j$ be the length-$k$ string such that $p_j[h]=1$ if $h=j$, and $p_j[h]=0$ otherwise.
  For each vertex $v_{j,i}$, let $q_{j,i}\in\{0,1\}^{m}$ be the string such that
  $$q_{j,i}[h] := \begin{cases}1,\; \text{if } 1\le h\le m \text{ and } v_{j,i}\in e_h\\ 0,\;\text{otherwise} \end{cases}$$
  and let $u_{j,i}:=p_jq_{j,i}$ be the concatenation of~$p_j$ and~$q_{j,i}$. 
  Note that $u_{j,i}$ has length $m':=m+k$, contains~$1+d$ ones and $m'-1-d$ zeros.   
  Let $\longZero:=0^{m'}$ be the string containing~$m'$ zeros and
  define the numbers
  \begin{align*}
    \kappa&:=knd+kn +k,\\
    \gamma&:=nk,\\
    \lambda&:=  \max\left\{ \left\lceil\kappa\left(\frac{2\mu}{\varepsilon}+1\right)\right\rceil, 2n(\gamma+k+1)\right\}+1.
  \end{align*}
Let $\ell:=\lambda(m'+1)\le\poly(nk)$.
For $1\leq j\leq k$, we define the string
$$s_j:=1 u_{j,1} (1 \longZero)^{\gamma+j}1 u_{j,2} (1 \longZero)^{\gamma+j}\ldots 1 u_{j,n} (1 \longZero)^{\gamma+j}(1 \longZero)^{\lambda-n(\gamma+j+1)}.$$
We further define the following \emph{dummy} string
$$s_0=11^{k}0^m(1 \longZero)^{\lambda-1}.$$
Note that each string~$s_j$ has length
$$n(m'+1)(1 +\gamma+j) + (m'+1)(\lambda - n(\gamma + j+1))=\lambda(m'+1)=\ell$$
Finally, we define the target cost
\begin{align*}
  c := &\ell f_{k+1}(0) 
  \\&  + \lambda(k+1)f'(k+1)
  \\& + 2\left(k+{k \choose 2}\right)(f'(2)-f'(1))
  \\& + \kappa f'(1).
\end{align*}
Clearly, the strings~$s_0,\ldots, s_{k}$ and the value~$c$ can be computed in polynomial time. This construction is illustrated in \Cref{fig:clique-to-fmscs}.

In the strings $s_0,\ldots,s_k$, any 1-value at a position~$i$ with $i\bmod (m'+1)=1$ is called a \emph{separator}, 
other 1-values are \emph{coding} positions. A coding position is either \emph{vertex-coding} if it belongs to some~$p_j$ 
(or to the~$k$ non-separator positions of~$s_0$), or \emph{edge-coding} otherwise (then it  belongs to some~$q_{i,j}$).
There are $\lambda (k+1)$ separator positions in total and $\kappa$ coding positions.

Given a multiple circular shift~$\Delta$, we define the \emph{weight} $w$ of a column as the number of 1-values it contains, 
that is, $0\leq w\leq k+1$. The cost for such column is  $f_{k+1}(w)=f_{k+1}(0)+wf'(w)$. Each 1-value of this column is
attributed a \emph{local cost} of $f'(w)$, so that the cost of any solution is composed of a \emph{base cost} of 
$\ell f_{k+1}(0)$ and of the sum of all local costs of all 1-values. In the following we mainly focus on local costs.

 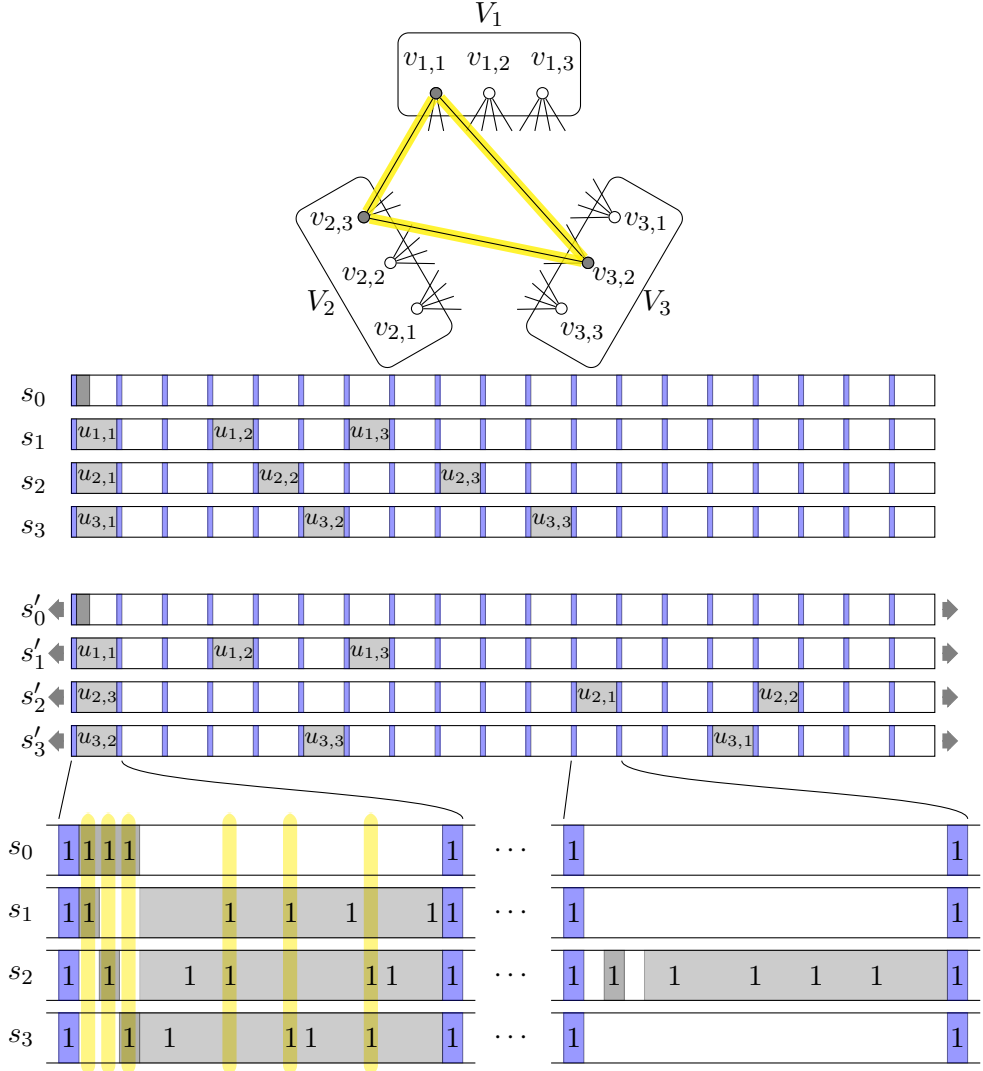
\begin{figure}
 \centering
 	\begin{tikzpicture}
	
 	\foreach \r [count=\i] in {90,210,330} {
		
 		\begin{scope} [rotate=\r]
 		\draw [rounded corners] (1.2,1.2) rectangle (2.3,-1.2);
 		\coordinate (V\i) at (2.55,0) ;
 		\coordinate (Vp\i) at (2,-1.5) ; 		
 		\foreach \y [count=\j] in {.7,0,-.7}{
 			\node[draw,circle, inner sep=1.5pt] (v\i\j) at (1.5,\y) {};
 			\node[] at (1.9,\y*1.2) {$v_{\i,\j}$};    
 			\foreach \delta [count=\d] in {-.3,-.1,.1,.3} {
 				\coordinate (n\i\j\d) at (1,\y+\delta);
 			}
 		}
		
 		\end{scope}
 	}
    \node at (V1) {$V_1$};
    \node at (V2) {$V_2$};
    \node at (V3) {$V_3$};
 	\foreach \s in {12,13,21,22,31,33} {
 		\foreach \d in {1,2,3,4}{
 			\draw (v\s) --(n\s\d);
 		}
 	}
 
 	\foreach \s/\d in {11/2,11/3,23/2,23/3,32/1,32/4} {
 		\draw (v\s) --(n\s\d);
 	}
 	\foreach \s in {11,23,32} {
		
 		\node[draw,circle, inner sep=1.5pt, fill=gray] at (v\s) {};
 	}
	
 	\foreach \s/\t in {11/23,11/32,23/32}{
 		\draw[line width=4pt, yellow, opacity=0.8] (v\s) -- (v\t);
 		\draw (v\s) -- (v\t);
 	}
 	\end{tikzpicture}
 	
 	\begin{tikzpicture}[scale=0.83,myRed/.style={fill=red!85!black,opacity=0.7}, myGray/.style={fill=black, opacity=0.2}, myBlue/.style={fill=blue, opacity=0.4}]
 	\def\w{.72}
	
 	\begin{scope}[shift={(0,4.5)}, yscale=0.7] %%% middle
	
         \foreach \i in {0,...,3} {
             \node at (-0.6,4-\i+0.2)  {$s_\i$};
             \draw (0,4-\i) rectangle (19*\w,4-\i+0.7);
             \foreach \j in {0, ...,18} {
                 \draw [myBlue] (\j*\w,4-\i) rectangle (\j*\w+0.08,4-\i+0.7);
             }            
         }
         \draw [myGray, opacity=0.4] (0.08,4) rectangle (0.4*\w,4+0.7);
     
         \foreach \i/\j/\k in {1/0/1,1/3/2,1/6/3, 2/0/1, 2/4/2,2/8/3, 3/0/1,3/5/2,3/10/3 } {
             \draw [myGray] (\j*\w+0.08,4-\i) rectangle ({(\j+1)*\w},4-\i+0.7);
             \node at (\j*\w+0.5*\w+.06,4-\i+0.35) {\footnotesize $u_{\i,\k}$};         
         }

 	\end{scope}

     \begin{scope}[shift={(0,1)}, yscale=0.7]
 	\foreach \i in {0,...,3} {
 		\node at (-0.6,4-\i+0.35)  {$s_\i'$};
 		\draw (0,4-\i) rectangle (19*\w,4-\i+0.7);
 		\foreach \j in {0, ...,18} {
 			\draw [myBlue] (\j*\w,4-\i) rectangle (\j*\w+0.08,4-\i+0.7);
 		}
		
 		\path[draw=black!50,solid,line width=2mm,fill=black!50,
 		preaction={-triangle 90,thin,draw=black!50,fill=black!50,shorten >=-1mm}
 		] (-0.12, 4.35-\i) -- (-0.25, 4.35-\i);
 		\path[draw=black!50,solid,line width=2mm,fill=black!50,
 		preaction={-triangle 90,thin,draw=black!50,fill=black!50,shorten >=-1mm}
 		] (19*\w+0.12, 4.35-\i) -- (19*\w+0.25, 4.35-\i);
 	}
 \draw [myGray, opacity=0.4] (0.08,4) rectangle (0.4*\w,4+0.7);
 
 	\foreach \i/\j/\k in {1/0/1,1/3/2,1/6/3, 2/0/3, 2/11/1,2/15/2, 3/14/1, 3/0/2,3/5/3 } {
 		\draw [myGray] (\j*\w+0.08,4-\i) rectangle ({(\j+1)*\w},4-\i+0.7);
 		\node at (\j*\w+0.5*\w+.06,4-\i+0.35) {\footnotesize $u_{\i,\k}$};         
		
 	}
	
 	\coordinate (zoom1) at (0,.9);
 	\coordinate (zoom2) at (\w+.08,.9);
	
 	\coordinate (zoom1b) at (11*\w+0,.9);
 	\coordinate (zoom2b) at (11*\w+\w+.08,.9);
 	\end{scope}
	
 	\def\w{.32}
 	\begin{scope}[shift={(-.2,-4.2)}]  %%%%%%%% bottom left
	
 	\coordinate (zoom3) at (0,4.9);
 	\coordinate (zoom4) at (17*\w+3*\w,4.9);
 	\draw (zoom1) -- (zoom3);
 	\draw (zoom2) to[in=90, out=-90, looseness=0.2] (zoom4);
	
 	\foreach \j in {-2,-1,0,5,8,12} {
 		\fill[yellow!60, rounded corners] (\j*\w+0.1*\w+3*\w,5) rectangle  (\j*\w+0.8*\w+3*\w,.8) ;
 	} 
	
 	\foreach \i in {0,...,3} {
 		\node at (-0.6,4-\i+0.4)  {$s_\i$};
 		%\draw (-0.2,4-\i) rectangle (16*\w,4-\i+0.8);
 		\draw (-0.2,4-\i) -- (17*\w+3*\w+.2,4-\i);
 		\draw (-0.2,4-\i+.8) -- (17*\w+3*\w+.2,4-\i+0.8);
 		\draw [myBlue] (0,4-\i) rectangle (\w,4-\i+0.8);  
 		\draw [myBlue] (16*\w+3*\w,4-\i) rectangle (16*\w+3*\w+\w,4-\i+0.8);  
 		\node at (16*\w+3*\w+.5*\w,4-\i+0.4) {$1$};
 		\node at (.5*\w,4-\i+0.4) {$1$};
 	}

 	\foreach \i in {1,...,3} {
 		\draw [myGray, opacity=0.2] (\w+3*\w,4-\i) rectangle (16*\w+3*\w,4.8-\i);
 	}
 	\foreach \j in {1,5,8,12,15} {
 		%\node at (.5*\w+\j*\w,4+0.4) {$1$};
 	}
 	\foreach \i/\j in {1/5,1/8,1/11,1/15,2/5,2/12,2/13,2/3, 3/2,3/8,3/12,3/9} {
 		\node at (.5*\w+\j*\w+3*\w,4-\i+0.4) {$1$};
 	}
 	\foreach \i/\j in {1/1,2/2,3/3} {
	\node at (.5*\w+\j*\w,4-\i+0.4) {$1$};
			\draw [myGray, opacity=0.3] (\j*\w,4-\i) rectangle (\w+\j*\w,4.8-\i);
}
 
	\foreach \s/\t/\i in {1/3/0} {
		\node at (.5*\w+\s*\w,4-\i+0.4) {$1$};
		\node at (.5*\w+\t*\w,4-\i+0.4) {$1$};
		\node at (.5*\w+\s*\w*.5+\t*\w*.5,4-\i+0.4) {$1$};%\tiny $\cdots$};
		\draw [myGray, opacity=0.3] (\s*\w,4-\i) rectangle (\w+\t*\w,4.8-\i);
	}
 	\foreach \j in {3.5,7,10.5,14} {
 	%	\node at  (\j*\w,4+0.4) {...};
 	}
	
 	\end{scope}
 	\begin{scope}[shift={(7.8,-4.2)}] %%%%%%%% bottom right
	
 	\coordinate (zoom3) at (0,4.9);
 	\coordinate (zoom4) at (17*\w+3*\w,4.9);
 	\draw (zoom1b) -- (zoom3);
 	\draw (zoom2b) to[in=90, out=-90, looseness=0.2] (zoom4);

 	\foreach \i in {0,...,3} {
 		\node at (-0.8,4-\i+0.4)  {$\cdots$};
 		%\draw (-0.2,4-\i) rectangle (16*\w,4-\i+0.8);
 		\draw (-0.2,4-\i) -- (17*\w+3*\w+.2,4-\i);
 		\draw (-0.2,4-\i+.8) -- (17*\w+3*\w+.2,4-\i+0.8);
 		\draw [myBlue] (0,4-\i) rectangle (\w,4-\i+0.8);  
 		\draw [myBlue] (16*\w+3*\w,4-\i) rectangle (16*\w+3*\w+\w,4-\i+0.8);  
 		\node at (16*\w+3*\w+.5*\w,4-\i+0.4) {$1$};
 		\node at (.5*\w,4-\i+0.4) {$1$};
 	}

 	\foreach \i in {2} {
 		\draw [myGray, opacity=0.2] (\w+3*\w,4-\i) rectangle (16*\w+3*\w,4.8-\i);
 	}
 	
 	\foreach \i/\j in {2/2,2/6,2/12,2/9} {
 		\node at (.5*\w+\j*\w+3*\w,4-\i+0.4) {$1$};
 	}
 	\foreach \i/\j in {2/2} {
	\node at (.5*\w+\j*\w,4-\i+0.4) {$1$};
			\draw [myGray, opacity=0.3] (\j*\w,4-\i) rectangle (\w+\j*\w,4.8-\i);
}
 	
 	\end{scope}
	
 	\end{tikzpicture}
 	\caption{\label{fig:clique-to-fmscs}Illustration of the reduction from an instance of \RMCC (top) with $k=3$. Middle: Sequences $s_0$ to $s_3$, and their optimal circular shifts $s_0'$ to $s_3'$. Blue stripes represent the regularly-spaced separator 1-values. The  (light) gray intervals contain both $0$'s and $1$'s according to strings~$u_{i,j}$, and white intervals contain only 0's. The spacing between consecutive $u_{i,j}$'s is defined using~$\gamma$ and the overall string length depends on $\lambda$, both values are chosen so as to restrict the possible alignments between different $u_{i,j}$'s; in this example we use $\gamma=2$ and~$\lambda=19$. Bottom: a zoom-in on blocks~1 and~12 in the shifted strings (only non-0 values are indicated, weight-2 columns are highlighted). Through vertex columns, the dummy string~$s_0$ ensures that one vertex occupies block~1 in each row, and weight-2 edge-columns ensure that~${k \choose 2}$ edges (as highlighted in the graph) are induced by these vertices.
 	}
 \end{figure}

It remains to show that there exists a multiple circular shift of $s_0,\ldots, s_{k}$ with cost~$c$ if~$G$ contains a properly colored $k$-clique, and that
otherwise every multiple circular shift has cost at least~$c+\varepsilon$.
We proceed by analyzing the structure and costs of optimal multiple circular shifts.

\paragraph{Aligning Separators.}

  Let $\Delta=(\delta_0,\ldots,\delta_{k})$ be a multiple circular shift of~$s_0,\ldots,s_{k}$.
  Without loss of generality, we can assume that $\delta_{0}=0$ since setting each~$\delta_j$ to~$(\delta_j-\delta_{0})\bmod \ell$ yields a shift with the same cost.
  First, we show that if $\delta_j \bmod (m'+1) \neq 0$ holds for some $0< j \le k$, then
  $\Delta$ has large cost.
  \begin{claim}\label{claim:blocks-aligned}
    For any multiple circular shift $\Delta=(\delta_0=0, \delta_1,\ldots,\delta_{k})$ with $\delta_j\bmod~(m'+1) \neq 0$ for some $1< j \le k$, it holds that $\cost_f(\Delta)\ge c+\varepsilon$.
  \end{claim}
  \begin{proof}  	
  	
    Assume that~$\delta_j\bmod (m'+1) = a\in\{1,\ldots,m'\}$ for some $0< j \le k$.  
    We count the number of weight-$(k+1)$ columns: such a column cannot only contain separator values since it cannot contain a separator value in both row~$0$ and row~$j$. Hence, it contains at least one coding value. Since there are $\kappa$ coding values, there are at most 
    $\kappa$ weight-$(k+1)$ columns, so at most $k\kappa$ separator values have local cost $f'(k+1)$. 
    All other separator values have local cost $f'(w)$ for some $w<k+1$, which is at least $f'(k+1)+\varepsilon$. There are at least  $\lambda (k+1)- k\kappa$ such separator values.
    Adding the base cost of $\ell f_{k+1}(0)$, the cost of $\Delta$ is thus at least:
    \begin{align*}
      \cost_f(\Delta) 
      &\geq \ell f_{k+1}(0) +  (\lambda(k+1)-k\kappa)(f'(k+1)+\varepsilon)\\
      &\geq \ell f_{k+1}(0) + \lambda(k+1)f'(k+1) + \lambda k\varepsilon - k\kappa(\mu+\varepsilon).      
     \end{align*}
    Recall that
    \begin{align*}
    c &= \ell f_{k+1}(0)  + \lambda(k+1)f'(k+1)
     + 2\left(k+{k \choose 2}\right)(f'(2)-f'(1))
     + \kappa f'(1) 
     \\
     &
     \leq \ell f_{k+1}(0) + \lambda (k+1) f'(k+1)
       +\kappa\mu
    \end{align*}
    since $f'(2)-f'(1) \le -\varepsilon < 0$.
    Combining the above bounds for $c$ and $\cost_f(\Delta)$ using $\lambda \ge \kappa\left(\frac{2\mu}{\varepsilon}+1\right)+1$ (by definition) yields
    \begin{align*}     
    \cost_f(\Delta) -c &\ge  \lambda k\varepsilon - k\kappa(\mu+\varepsilon) - 
                         \kappa\mu\\
      &\ge 2\kappa k \mu + \kappa k \varepsilon + k\varepsilon - k\kappa(\mu+\varepsilon) - 
                         \kappa\mu\\
    &\geq \varepsilon.
    \end{align*}
  \end{proof}

\paragraph{Cost of Circular Shifts.}
We assume from now on that $\delta_j\bmod (m'+1)=0$ for all $j\in\{0,\ldots,k\}$. We now provide a precise characterization of the cost of~$\Delta$.

For~$l\in\{1,\ldots,\lambda\}$, we define the \emph{$l$-th block} consisting of the~$m'$ consecutive columns~$(l-1)(m'+1)+2,\ldots,l(m'+1)$. The \emph{block index} of column $i$ is $i-1\bmod (m'+1)$.
For $j\in\{1,\ldots,k\}$, the substring $s_j^{\leftarrow\delta_j}[(l-1)(m'+1)+2]\ldots s_j^{\leftarrow\delta_j}[l(m'+1)]$ corresponding to the $l$-th block of~$s_j^{\leftarrow\delta_j}$
either equals some $u_{j,i}$ or $\longZero$.
We say that block~$l$ is \emph{occupied} by vertex $v_{j,i}\in V_j$, %~$j\in\{1,\ldots,k\}$
if the corresponding substring of~$s_j^{\leftarrow\delta_j}$ is $u_{j,i}$. Note that for each $j$ there are $n$ distinct blocks out of $\lambda$ that are occupied by a vertex in $V_j$. Columns with block-index $1$ to $k$ are called \emph{vertex-columns}  and columns with block-index $k+1$ to $k+m=m'$ are \emph{edge-columns} (they may only contain edge-coding values from some $q_{i,j}$). 
Let~$P$ denote the set of vertices occupying block~$1$.

\begin{observation}\label{obs:vertex-column}
In block $l$, if the vertex-column with block-index $h$ has weight~2, then $l=1$, and $V_h\cap P\neq\emptyset$. No vertex-column can have weight 3 or more.
\end{observation}
\begin{proof}
  Consider the vertex-column with block-index $h$. By construction, only $s_h$ may have a 1 in this column (which is true if some vertex from $V_h$ occupies this block). The string~$s_0$ has a 1 in this column if it is a column in block 1.
  Thus, assuming that column~$h$ has weight~2 implies $l=1$ and $V_h\cap P\neq \emptyset$. 
\end{proof}
\begin{observation}\label{obs:edge-column}
In block $l$, if the edge-column with block-index $k+h$, $1\leq h\leq m$, has weight 2, then block~$l$ is occupied by both vertices of edge~$e_h\in E$. No edge-column can have weight 3 or more.
\end{observation}
\begin{proof}
Consider an edge-column with block-index $k+h$, $1\leq h\leq m$. Denote by $v_{j_0,i_0}$ and $v_{j_1,i_1}$ the endpoints of edge $e_h$.
For any $1\leq j\leq k$, $s_j$ has a 1 in this column only if block $l$ is occupied by some vertex $v_{j,i}$, and, moreover, only if $u_{j,i}$ has a 1 in column $h$, i.e. $v_{j,i}= v_{j_0,i_0}$ or $v_{j,i}= v_{j_1,i_1}$, hence $j=j_1$ or $j=j_2$. So this column may not have weight 3 or more, and if it has weight~2,  then block~$l$ is occupied by both endpoints of $e_h$.
\end{proof}

From Observations~\ref{obs:vertex-column} and~\ref{obs:edge-column}, it follows that no column (beside separators) can have weight~3 or more. Since the number of coding values is fixed, the cost is entirely determined by the number of weight-2 columns. The following result quantifies this observation.

\begin{claim}\label{claim:costW2}
  Let $W_2$ be the number of weight-2 columns.
  If~$W_2=k+{k \choose 2}$, then $\cost_f(\Delta) = c$. 
  If~$W_2 < k+{k \choose 2}$, then $\cost_f(\Delta)\ge c + \varepsilon$.
\end{claim}
\begin{proof}
	The base cost $\ell f_{k+1}(0)$ of the solution only depends on the number~$\ell$ of columns.
	Separator values are in weight-$(k+1)$ columns.
        Since there are~$\lambda$ of them, it follows that the total local cost of all separator values is $\lambda(k+1)f'(k+1)$.
	
	The total number of coding values is  $\kappa$, each coding value has a local weight of $f'(1)$ if it belongs to a weight-1 column, and $f'(2)$ otherwise (since there is no vertex- or edge-column with weight 3 or more). There are $W_2$ weight-2 columns, so exactly $2W_2$ coding values within weight-2 columns. Summing the base cost with the local costs of all separator and coding values, we get:
		\begin{align*}
		 \cost_f(\Delta) 
		 = & \ell f_{k+1}(0) 
		 \\&  + \lambda(k+1)f'(k+1)
		 \\& + 2W_2(f'(2)-f'(1))
		 \\& + \kappa f'(1).
		\end{align*}
	
                Thus, by definition of~$c$, we have $\cost_f(\Delta) = c$ if~$W_2 = k+{k \choose 2}$.
                If $W_2 < k + {k\choose 2}$, then using the fact that, by assumption, $f'(2)-f'(1) \le -\varepsilon$, we obtain
                $$\cost_f(\Delta) = c + 2\left(W_2-k-{k\choose 2}\right)(f'(2)-f'(1)) \ge c+\varepsilon.$$
\end{proof}

Since the cost is determined by the number of weight-2 columns, we need to evaluate this number. \Cref{obs:vertex-column} gives a direct upper bound for weight-2 vertex columns (at most~$k$, since they all are in block~1), hence we now focus on weight-2 edge-columns. The following claim will help us to upper-bound their number.

\begin{claim} \label{claim:noSquareCorners}
	For any two rows $j,j'$, there exists at most one block $l$ that contains vertices from both $V_j$ and $V_{j'}$.
\end{claim}
\begin{proof}
	If two distinct blocks $l,l'$ contain vertices from the same row $j$, then two cases are possible: either $|l-l'|= a(\gamma+j+1)$ or  $|l-l'|= \lambda-a(\gamma+j+1)$, in both cases with $1\leq a\leq n$. Indeed, there are $n$ regularly-spaced substrings $u_{j,i}$ in row $j$, so the two cases correspond to whether or not the circular shifting of row $j$ separates these two blocks.
	
	Aiming at a contradiction, assume that two distinct rows $j$ and $j'$ provide two vertices for  both $l$ and $l'$. Then there exist $1\leq a,a'\leq n$ such that $|l-l'|= a(\gamma+j'+1)$ or  $|l-l'|= \lambda-a(\gamma+j+1)$, and $|l-l'|= a'(\gamma+j'+1)$ or  $|l-l'|= \lambda-a'(\gamma+j'+1)$. This gives four cases to consider (in fact just three by symmetry of $j$ and $j'$).
	
	If $|l-l'| = a(\gamma+j+1) = a'(\gamma+j'+1)$, then $(a-a')(\gamma+1) = a'j'-aj$. We have $a\neq a'$, as otherwise this would imply $j=j'$. So $|a'j'-aj|\geq \gamma+1$, but this is impossible since $a,a'\leq n$ $j,j'\leq k$, and $\gamma>kn$ by construction. 
	
	If $|l-l'| = a(\gamma+j+1) = \lambda-a'(\gamma+j'+1)$, then $\lambda = a(\gamma+j+1)+a'(\gamma+j'+1)$. However, $\lambda>2n(\gamma+k+1)$ by construction, so this case also leads to a contradiction.
	
	Finally, if $|l-l'| = \lambda - a(\gamma+j+1) = \lambda-a'(\gamma+j'+1)$, then we have $a(\gamma+j+1) = a'(\gamma+j'+1)$. This case yields, as in the first case, a contradiction.
\end{proof}

\begin{claim}\label{claim:fewHeavyEdgeColumns}
	There are at most ${k \choose 2}$ weight-2 edge-columns.
\end{claim}
\begin{proof}
	Consider any pair $j,j'$ such that $1\leq j<j'\leq k$. It suffices to show that there exists at most one weight-2 edge-column with a 1 in rows~$j$ and~$j'$. 
	
	Aiming at a contradiction, assume that two such columns exist. By Observation~\ref{obs:edge-column}, they must each belong to a block which is occupied by vertices both in $V_j$ and $V_{j'}$. From \Cref{claim:noSquareCorners} it follows that both columns belong to the same block. Let $v$ and $v'$ be the vertices of $V_j$ and~$V_{j'}$, respectively, occupying this block. By Observation~\ref{obs:edge-column} again, both edges are equal to $\{v,v'\}$, which contradicts the fact that they are distinct.	
\end{proof}

\begin{claim}\label{claim:noCliqueManyW2}
	If $G$ does not contain a properly colored $k$-clique, then there are at most  $k+{k \choose 2}-1$ weight-2 columns.
\end{claim}
\begin{proof}
	Assume that there are at least $k+{k \choose 2}$ weight-2 columns. By \Cref{claim:fewHeavyEdgeColumns}, there are at least~$k$ weight-2 vertex-columns. By \Cref{obs:vertex-column}, only the $k$  vertex-columns of block 1 may have weight 2, hence for each $1\leq j\leq k$ the column of block~1 with block-index $j$ has weight 2. Thus for every, $j$, $P\cap V_j\neq \emptyset$. 
	
	By \Cref{claim:noSquareCorners}, no other block than block 1 may be occupied by two vertices, hence any edge-column with weight 2 must be in block~1, and both endpoints are in $P$. There cannot be more than $k$ weight-2 vertex-columns, hence there are ${k \choose 2}$ weight-2 edge-columns, and for each of these there exists a distinct edge with both endpoints in $P$. 
	Thus, $P$ is a properly colored $k$-clique.
\end{proof}

\paragraph{Cliques and Circular Shifts with Low Cost.}

We are now ready to complete the proof of \Cref{thm:MSCS_reduction}.
First, assume that $G$ contains a properly colored $k$-clique~$P=\{v_{1,i_1},\ldots,v_{k,i_k}\}$.
Consider the multiple circular shift~$\Delta=(\delta_0,\ldots,\delta_{k})$, where
$\delta_{0}=0$ and
$$\delta_j:=(i_j-1)(m'+1)(\gamma + j + 1)$$ for~$j\in\{1,\ldots,k\}$.
Note that~$|P|=k$, and all edge-columns in block~1 corresponding to edges induced in $P$  have weight~2. Hence there are $\binom{k}{2}$ weight-2 edge-columns and $k$ weight-2 vertex-columns.
By \Cref{claim:costW2}, $\cost_f(\Delta) = c$.

Now, assume that~$G$ does not contain a properly colored $k$-clique.
Without loss of generality, let~$\Delta=(\delta_0,\ldots,\delta_k)$ be a multiple circular shift
with~$\delta_{0}=0$.
Clearly, if~$\delta_j\bmod (m'+1)\neq 0$ for some~$j$, then~$\cost_f(\Delta)\ge c+\varepsilon$ (by \Cref{claim:blocks-aligned}).
Otherwise, by~\Cref{claim:noCliqueManyW2} there are at most $k+\binom{k}{2}-1$ weight-2 columns. By \Cref{claim:costW2} , $\cost_f(\Delta) \geq c+\varepsilon$.

This completes the proof of \Cref{thm:MSCS_reduction} which directly leads to our main result of this section.

%\end{proof}

\begin{theorem}\label{cor:MSCShardness}
  Let~$f$ be a polynomially bounded grouping function.
  Then, \fMSCS on binary strings is
  \begin{enumerate}[(i)]
    \item\label{NPh} NP-hard,
    \item\label{W1h} W[1]-hard with respect to the number~$k$ of input strings, and
    \item\label{ETH} not solvable in $\rho(k)\cdot n^{o(k)}$ time for any computable function~$\rho$ unless the ETH fails.
    \end{enumerate}
\end{theorem}

\begin{proof}
  The polynomial-time reduction from \Cref{thm:MSCS_reduction} yields the NP-hardness.
  Moreover, the number of strings in the \fMSCS instance only depends on the size of the multicolored clique. Hence, it is a parameterized reduction from \RMCC parameterized by the size of the clique to \fMSCS parameterized by the number of input strings and thus yields W[1]-hardness.
  Lastly, the number~$k'=k+1$ of strings is linear in the size~$k$ of the clique. Thus, any $\rho(k')\cdot n^{o(k')}$-time algorithm for \DTW would imply a $\rho'(k)\cdot |V|^{o(k)}$-time algorithm for \RMCC contradicting the ETH.
\end{proof}

\noindent Note that \Cref{cor:MSCShardness} holds for the function~$\sigma$ since it is
a polynomially bounded grouping function (as discussed earlier).

The assumption that~$f$ is polynomially bounded is only needed to obtain a polynomial-time reduction in~\Cref{thm:MSCS_reduction}.
Without this assumption, we still obtain a parameterized reduction from \RMCC parameterized by the clique size to \fMSCS parameterized by the number of input strings, which yields the following corollary for a larger class of functions.

\begin{corollary}
  Let~$f$ be a computable grouping function.
  Then, \fMSCS on binary strings is W[1]-hard with respect to the number~$k$ of input strings and not solvable in $\rho(k)\cdot n^{o(k)}$ time for any computable function~$\rho$ unless the ETH fails.
\end{corollary}

\section{Circular Consensus String}\label{sec:ccs}
In this section we briefly study the \textsc{Circular Consensus String} (CCS) problem: Given $k$ strings $s_1,\ldots, s_k$ of length~$n$ each, find a length-$n$ string $s^*$ and a circular shift $(\delta_1,\ldots,\delta_k)$ such that $\sum_{j=1}^kd(s_j^{\leftarrow \delta_j}, s^*)$ is minimal, where $d$ denotes the Hamming distance, that is, the number of mismatches between the positions of two strings.
Although consensus string problems in general have been widely studied from a theoretical point of view~\cite{BHKN14},
somewhat surprisingly this is not the case for the circular version(s).
For CCS, only an $O(n^2\log n)$-time algorithm for~$k=3$ and an~$O(n^3\log n)$-time algorithm for~$k=4$ is known~\cite{LNPPS13}.
However, for general~$k$ no hardness result is known.
Note that without circular shifts the problem is solvable in linear time:
It is optimal to set $s^*[i]$ to any element that appears a maximum number of times among the elements $s_1[i],\ldots,s_k[i]$.

For binary strings, it can easily be seen that the cost induced by column~$i$ is the minimum of the number of $0$'s and the number of $1$'s.
Let $\fCCS$ be the polynomially bounded order-independent function with $\fCCS[k](w)=\min\{w, k-w\}$.
It follows from the discussion above that \textsc{Circular Consensus String} is exactly \fMSCS[\fCCS]. Note, however, that~$\fCCS$ is not a grouping function since $\fCCS[k]'(2)=\fCCS[k]'(1)=1$. That is, we do not immediately obtain hardness of CCS from \Cref{cor:MSCShardness}.
We can still prove hardness via a reduction using a properly chosen polynomially bounded grouping function.

\begin{theorem}\label{thm:CCS-hard}
  \textsc{Circular Consensus String} on binary strings is
  \begin{enumerate}[(i)]
    \item NP-hard,
    \item W[1]-hard with respect to the number~$k$ of input strings, and
    \item not solvable in~$\rho(k)\cdot n^{o(k)}$ time for any computable function~$\rho$ unless the ETH fails.
  \end{enumerate}
\end{theorem}
\begin{proof}
  As discussed above, CCS is equivalent to \fMSCS[\fCCS].
  To prove hardness, we define a local cost function $g$ (similar to $\fCCS$) and reduce from \fMSCS[$g$] to \fMSCS[$\fCCS$].

  Let $g$ be the order-independent local cost function such that $$g_k(w) :=\fCCS[2k-2](w+(k-2))=\min\{w+k-2, k-w\}.$$
  Note that the function $g_k$ is linearly decreasing on $\{1,\ldots, k\}$ and that $g'_k(w)=\frac{2-w}w=\frac2w -1$.
  The range of~$g_k$ is $\mu_k=1$ and its gap is $\varepsilon_k=\frac2{k-1} - \frac2k >\frac2{k^2}$. That is, $g$ satisfies all conditions of \Cref{cor:MSCShardness} and the corresponding hardness results hold for \fMSCS[$g$].
  See \Cref{fig:red-gMSCS} for an illustration.

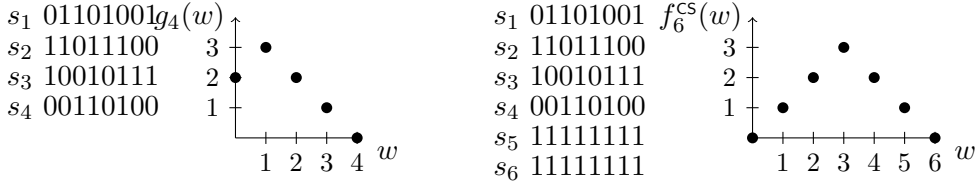
\begin{figure}
\begin{tikzpicture}[scale=0.8]
\coordinate (A) at (0,0);
\coordinate (B) at (8,0);
\coordinate (C) at (2.5,-2.5);
\coordinate (D) at (11,-2.5);
\begin{scope}[yscale=0.5]
 \foreach \s [count=\i] in {$s_1$ 01101001, $s_2$ 11011100, $s_3$ 10010111, $s_4$ 00110100} {
	\node at ($(A)+(0,-\i)$) {\s};	
	\node at ($(B)+(0,-\i)$) {\s};
 }
 \foreach \i in {5,6} {
	\node at ($(B)+(0,-\i)$) {$s_\i$ 11111111};
 }
 \end{scope}
 
\foreach \C/\l/\w in {C/g_4/2.0, D/\fCCS[6]/3.0}{
	\draw[->] (\C)--++(0,2) node[left] {$\l(w)$}; 
	\draw (\C)--++(\w,0);
        \node at ($(\C) + (\w+0.5,-0.25)$) {$w$};
      }

      \begin{scope}[xscale=0.5, yscale=0.5]
        \foreach \y in {1,2,3} {
          \draw ($(C)+(-0.2,\y)$) node[left] {\small\y} -- ($(C)+(0.2,\y)$);
          \draw ($(D)+(-0.2,\y)$) node[left] {\small\y} -- ($(D)+(0.2,\y)$);
        }
        \foreach \x in {1,2,3,4} {
          \draw ($(C)+(\x,0.2)$) -- ($(C)+(\x,-0.2)$) node[below] {\small\x};
        }
        \foreach \x in {1,2,3,4,5,6} {
          \draw ($(D)+(\x,0.2)$) -- ($(D)+(\x,-0.2)$) node[below] {\small\x};
        }
        \foreach \x/\y [ remember=\y as \lasty (initially 2),remember=\x as \lastx (initially 0)] in {0/2,1/3,2/2,3/1,4/0} {
          \node[circle,fill=black,scale=0.4] at ($(C)+(\x,\y)$) {};
        }
        \foreach \x/\y [ remember=\y as \lasty (initially 0),remember=\x as \lastx (initially 0)] in {0/0,1/1,2/2,3/3,4/2,5/1,6/0} {
          \node[circle,fill=black,scale=0.4] at ($(D)+(\x,\y)$) {};
        }
      \end{scope}
\end{tikzpicture}
\caption{\label{fig:red-gMSCS}Reduction from an instance of \fMSCS[$g$] (left) to an instance of \fMSCS[\fCCS], which is equivalent to the {\sc Circular Consensus String} problem. Plots of the (polynomially bounded and order-independent) local cost functions for $k=4$ are shown. Note that~$g_4$ is obtained from~$\fCCS[6]$ by cropping the first two values in order to become grouping.}
\end{figure}

Given an instance $\mathcal I = (s_1,\ldots, s_k, c)$ of \fMSCS[$g$], we define the strings $s_j:=1^{|s_1|}$ for $j=k+1,\ldots,2k-2$. We show that~$\mathcal I$ is a yes-instance if and only if $\mathcal I' := (s_1,\ldots, s_{2k-2}, c)$ is a yes-instance for \fMSCS[\fCCS]. 

For the forward direction, consider a multiple circular shift $\Delta=(\delta_1,\ldots,\delta_k)$ of $s_1,\ldots,s_k$ such that $\cost_g(\Delta)\leq c$. We define the multiple circular shift $\Delta':=(\delta_1,\ldots,\delta_k,\delta_{k+1}=0,\ldots,\delta_{2k-2}=0)$ of~$s_1,\ldots,s_{2k-2}$. Consider column~$i$ of~$\Delta'$ and let~$w'$ be the number of 1's it contains.
Then, $w'=w+k-2$, where $w$ is the number of 1's in column~$i$ of~$\Delta$. The cost of column~$i$ is $\fCCS[2k-2](w+k-2) = g_k(w)$. Hence, column~$i$ has the same cost in both solutions.
This implies~$\cost_g(\Delta) = \cost_{\fCCS}(\Delta')$.

The converse direction is similar. Any multiple circular shift $\Delta'$ of $s_1,\ldots,s_{2k-2}$ can be restricted to a multiple circular shift $\Delta$ of $s_1,\ldots,s_k$ with the same cost.
\end{proof}

\section{Consensus for Time Series: \DTW}\label{sec:dtw}

In this section we prove the following theorem,
settling the complexity status of a prominent consensus 
problem in time series analysis. 

\begin{theorem}\label{thm:DTW-hard}
\DTW on binary time series is
  \begin{enumerate}[(i)]
    \item NP-hard,
    \item W[1]-hard with respect to the number~$k$ of input series, and
    \item not solvable in $\rho(k)\cdot n^{o(k)}$ time for any computable function~$\rho$ unless the ETH fails.
    \end{enumerate}
\end{theorem}

The proof is based on a polynomial-time reduction from a special variant of~\fMSCS for which hardness holds via \Cref{cor:MSCShardness} in \Cref{sec:mscs}.
At this point we make crucial use of the fact that the reduction from the proof of \Cref{thm:MSCS_reduction} actually shows that it is hard to decide whether there is a multiple circular shift of cost at most~$c$ or whether all multiple circular shifts have cost at least~$c+\varepsilon$ for some (polynomially bounded)~$\varepsilon$. This guarantees that a no-instance of \fMSCS is reduced to a no-instance of \DTW.

Before giving the proof, we introduce some definitions.
A \emph{position} $p$ in a sequence $x$ is an integer $1\leq p\leq |x|$, its \emph{value} is $x[p]$. The \emph{distance} between two positions $p$ and $p'$ is $|p'-p|$.
A \emph{block} in a binary sequence is a maximal subsequence of 0's (a \emph{0-block}) or 1's (a \emph{1-block}). Blocks are also represented by integers, indicating their rank in the sequence (a sequence with $n$ blocks has blocks $1,2,\ldots,n$). The \emph{distance} between two blocks of rank $y$ and $y'$ is $|y'-y|$. Note that the notion of distance is different in the context of positions and blocks (even between size-1 blocks, as larger blocks in between increase the position distance).
% In the following proof, we mostly consider block distances for input time series and position distances for the mean.

\begin{proof} We will reduce from a special variant of \fMSCS.
  To this end, consider the function $f_k \colon \{0,\ldots,k\} \rightarrow [0,1]$ with
        \[f_k(x) =\frac {k+x}{k+x+1}\] and $\phi_k \colon \{0,\ldots,k\} \rightarrow [0,1]$ with \[\phi_k(x)=f_k(x)-f_k(0)=\frac{x}{(k+x+1)(k+1)}.\]
        Note that~$\phi \colon \{0,1\}^*\rightarrow [0,1]$ with~$\phi((x_1,\ldots,x_k)) = \phi_k(\sum_{j=1}^kx_j)$ is a polynomially bounded grouping function since
        \[\phi_k'(x) = \frac{1}{(k+x+1)(k+1)}\]
        is strictly decreasing with gap~$\epsilon_k=\frac{1}{(k+1)(2k)(2k+1)}$ and range~$\mu_k=\frac{1}{(k+2)(k+1)}$ (see \Cref{fig:phi} for an example).
        Hence, by \Cref{cor:MSCShardness} in \Cref{sec:mscs} hardness holds for \fMSCS[$\phi$].
        We give a polynomial-time reduction from \fMSCS[$\phi$].
\paragraph{Reduction.}

        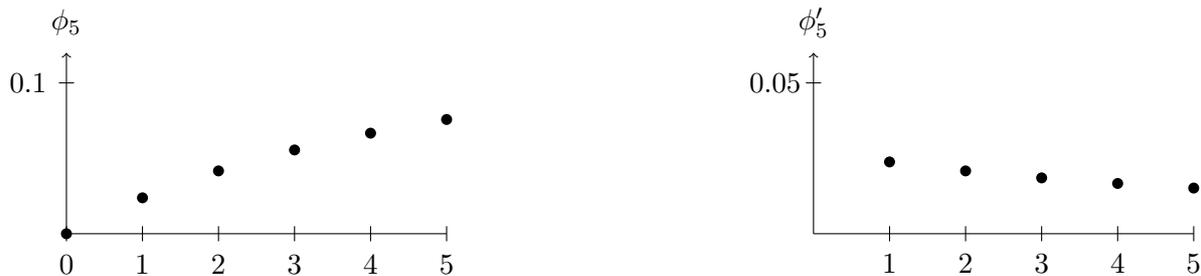
\begin{figure}[t]
          \centering
          \begin{tikzpicture}[yscale=20.0]
            \draw (0,0) -- (5,0);
            \draw[->] (0,0) -- (0,0.12);
            \draw (-0.1,0.1) -- (0.1,0.1);
            \node at (-0.5,0.1) {0.1};
            \node at (0,0.14) {$\phi_5$};
            \foreach \x in {0,1,...,5} {
              %\pgfmathsetmacro\result{(5.0+\x)/(6.0+\x)}
              %\node[circle,fill=black,scale=0.3] at (\x,\result) {};
              \pgfmathsetmacro\result{\x/(36+6*\x)}
              \node[circle,fill=black,scale=0.4] at (\x,\result) {};
              \draw (\x,-0.005) -- (\x,0.005);
              \node at (\x,-0.02) {\x};
            }
          \end{tikzpicture}
          \hfill
          \begin{tikzpicture}[yscale=40.0]
            \draw (0,0) -- (5,0);
            \draw[->] (0,0) -- (0,0.06);
            \draw (-0.1,0.05) -- (0.1,0.05);
            \node at (-0.5,0.05) {0.05};
            \node at (0,0.07) {$\phi_5'$};
            \foreach \x in {1,...,5} {
              \pgfmathsetmacro\result{1 / (36+6*\x)}
              \node[circle,fill=black,scale=0.4] at (\x,\result) {};
              \draw (\x,-0.0025) -- (\x,0.0025);
              \node at (\x,-0.01) {\x};
            }
          \end{tikzpicture}
          \caption{Left: The function~$\phi_5(x) = \frac{x}{6(6+x)}$. Right: The function $\phi_5'(x)=\frac{1}{6(6+x)}$.}
          \label{fig:phi}
        \end{figure}

	In the following, we assume to have an instance of \fMSCS[$\phi$] with $k\ge 15$ length-$n$ binary strings $s_1,\ldots,s_k\in\{A,B\}^n$ (where~$A:=0$ and~$B:=1$) and a target cost~$c$.
        We write~$\epsilon$ for the gap~$\epsilon_k$ of~$\phi_k$ (note that $\epsilon^{-1}\in O(k^3)$).
        The task is to decide whether there exists a multiple circular shift of cost at most $c$ or whether all multiple circular shifts have cost at least $c+\epsilon$ (this is proved to be hard by \Cref{thm:MSCS_reduction}.)
	
	We define the numbers
        \begin{align*}
          m \coloneqq   \left\lceil 1600k(c+\epsilon)\right\rceil \text{ and }
          r \coloneqq \left\lceil \frac{1}\epsilon (3mnk + 2 (c+\epsilon)) \right\rceil +1.
        \end{align*}
	
	Furthermore, we define the binary strings
        $t_A\coloneqq (10)^{m}$ and $t_B\coloneqq 100(10)^{m-1}$, each having~$m$ 0-blocks (all of length one, except for the first 0-block of~$t_B$ which has length two).
	The first 0-block of~$t_A$ and~$t_B$ is called a \emph{coding} block (respectively, an $A$-coding or $B$-coding block).
	
	Let $s'_j$ be the string obtained by concatenating the strings $t_{s_j[i]}$ for $1\leq i\leq |s_j|$. Let $x_j$ be the time series obtained by concatenating~$r$ copies of~$s'_j$.
        Note that each~$x_j$ contains $2mnr$ blocks.
	We also define the \emph{extra} series $x_{k+1}=(1)$ and set the target cost to
        $$c'\coloneqq r(c+ mn f_k(0)) + 3mnk.$$
	
	For the correctness of this reduction we need to prove the following:
	\begin{enumerate}[(i)]
		\item If $(s_1,\ldots,s_k,c)$ is a yes-instance of \fMSCS[$\phi$] (that is, optimal cost at most~$c$), then there exists a mean~$z$ of~$x_1,\ldots,x_{k+1}$ with $\mathcal{F}(z)\le c'$.
		\item If $(s_1,\ldots,s_k,c)$ is a no-instance of \fMSCS[$\phi$] (that is, optimal cost at least $c+\epsilon$), then $\mathcal{F}(z) > c'$ holds for every time series~$z$.
	\end{enumerate}

        \paragraph{Yes-instance of \fMSCS[$\phi$].}
	Consider a multiple circular shift $\Delta=(\delta_1,\ldots,\delta_k)$ of  $s_1,\ldots,s_k$ with cost at most~$c$. Without loss of generality, we assume that $0\le \delta_j < n$ holds for every~$1\le j \le k$.  
	
	We construct a mean~$z$ of length $2mn(r-1)+2$ as follows:
        The first position is aligned with the leftmost $2\delta_jm$ blocks of each $x_j$ (or to the first block if $\delta_j=0$) and the last position is aligned with the rightmost $2(n-\delta_j)m$ blocks of $x_j$.
        For each $1 < i < 2mn(r-1)+2$, position~$i$ is aligned with the $(i-1)+(2\delta_jm)$-th block of~$x_j$, for each $x_j$.
        These positions are called \emph{regular} positions, whereas the
        first and last position are called \emph{extreme}.
        Clearly, all positions of~$z$ are also aligned with the single~1 in $x_{k+1}$.
	
	\begin{claim}\label{clm:reg}
	The total cost of regular positions is at most  $mn(r-1)f_k(0) + (r-1)c$.	
	\end{claim}
	\begin{proof}
          Due to the alternation of 1- and 0-blocks in each~$x_j$ and the fact that $i+(2\delta_jm)\equiv i \pmod 2$, it follows that the $i$-th regular position is mapped only to~1's if $i$ is odd (\emph{odd} position) or only to 0's (except for the single 1 in $x_{k+1}$) if $i$ is even (\emph{even} positions). Thus, odd positions have cost~0, and even positions have a cost depending on the size of the 0-blocks to which they are mapped.
	
          Consider an even position~$i$ such that $i\bmod 2m \neq 2$.
          The $i$-th regular position is not mapped to any coding position in any $x_j$.
          Thus, it is mapped to~$k$ 0's and a single~1, and has cost $\frac{k}{k+1}=f_k(0)$. There are $(m-1)n(r-1)$ such positions, which thus contribute for a total cost of $(m-1)n(r-1)f_k(0)$.
	
          For an even position~$i$ with $i\bmod 2m = 2$, the $i$-th regular position is mapped to a coding block in each $x_j$ (except for the single 1 in $x_{k+1}$). Let $i=2mi'+2$.
   Then, $z[i]$ is mapped to coding positions corresponding to column $i' \bmod n$ of $\Delta$.
   If this column contains $a$ $A$'s and $k-a$ $B$'s, then~$z[i]$ is mapped to $a+2(k-a)$ 0's and a single~1 and has cost
   $$\frac{a+2(k-a)}{a+2(k-a)+1} = 1-\frac 1{2k-a+1} = f_k(k-a)=f_k(0)+\phi_k(k-a).$$
   Note that column $i' \bmod n$ of~$\Delta$ has cost $\phi_k(k-a)$.
   Hence, the overall cost of the $(r-1)n$ regular positions~$i$ with~$i\bmod 2m = 2$ is $$(r-1)\cost_\phi(\Delta)+(r-1)nf_k(0) \le (r-1)c+(r-1)nf_k(0).$$

   Overall, regular positions have cost at most
   $$(m-1)n(r-1)f_k(0) + n(r-1)f_k(0) + (r-1)c = mn(r-1)f_k(0)+(r-1)c.$$	
 \end{proof}

 \begin{claim}\label{clm:extreme}
   The total cost of extreme positions is at most $2knm+2$.
 \end{claim}
 \begin{proof}
   Aiming at an upper bound, assume that the extreme positions of~$z$ have value 0. 
   Since $0\le \delta_j< n$, an extreme position is mapped to at most $2nm$ consecutive blocks in each $x_j$, thus accounting for at most $nm$ 1's in $x_j$ plus an additional 1 in $x_{k+1}$, yielding an upper bound of cost $knm+1$ for each extreme position.
 \end{proof}
 
 Combining \Cref{clm:reg,clm:extreme}, we obtain
 \begin{align*}
   \mathcal{F}(z) &\le mn(r-1)f_k(0) + (r-1)c + 2knm+2 \\
   &\le r(c+mn f_k(0)) + 2knm + 2 -mnf_k(0)  \\
   &= r(c+mn f_k(0)) + mn(2k+f_k(0)) + 2   \\	   
   &\leq r(c+mn f_k(0)) + mn(2k+1) + 2 \\
   &\leq r(c+mn f_k(0)) + 3mnk = c'.
 \end{align*}
Overall, $z$ is a mean for $(x_1,\ldots,x_{k+1})$ with cost at most $c'$, so $(x_1,\ldots,x_{k+1}, c')$ is a yes-instance of $\DTW$.

\paragraph{No-instance of \fMSCS[$\phi$].}
	We assume that the \fMSCS[$\phi$] instance~$(s_1,\ldots,s_k,c)$ is a no-instance, and we consider a fixed mean time series $z$ together with optimal warping paths between $z$ and $x_1,\ldots, x_{k+1}$. We aim at proving a lower bound on the cost of $z$.
	We say that $x_j[i]$ is \emph{matched} to~$z[i']$ if~$(i,i')$ is in the warping path between~$x_j$ and~$z$.
        Clearly, every position of $z$ is matched to the single position in~$x_{k+1}$.

        We write $\#_1(p)$ and $\#_0(p)$ respectively for the number of positions with value 1 (resp. 0) to which $z[p]$ is matched among $x_1,\ldots, x_k$ (ignoring the matching with the extra sequence $x_{k+1}$). The cost of $z[p]$ is
        $$C(p)\coloneqq\frac{\#_0(p)(\#_1(p)+1)}{\#_0(p)+\#_1(p)+1}.$$

        We will use the following monotonicity property of the cost.

\begin{lemma}\label{lem:monotmonousCost}
For any $a\geq a'\geq 0$ and $b\geq b' \geq 1$, it holds $\frac{ab}{a+b} \geq \frac{a'b'}{a'+b'}$.
\end{lemma}
\begin{proof}
  It suffices to see that the partial derivatives
  $$\frac{\partial}{\partial a } \frac{ab}{a+b} = \frac{b^2}{(a+b)^2} \quad\text{and}\quad\frac{\partial}{\partial b} \frac{ab}{a+b} = \frac{a^2}{(a+b)^2}$$ 
  are non-negative for~$a\ge 0$ and $b\ge 1$.
\end{proof}	

We define some further notation for the remaining part of the proof.
The \emph{range} of a position~$p$ of~$z$ in $x_j$ is the set of positions of $x_j$ to which $p$ is matched. The range is a subinterval of $\{1,2,...,|x_j|\}$ and by construction of $x_j$, its values may not have three consecutive 0's or two consecutive 1's.
More precisely, its values alternate between 0 and 1, except for the (rare) occasions where it includes a $B$-coding block. 
The number of blocks of~$x_j$ intersecting the range of $p$ in $x_j$ is denoted $r_j(p)$. 
The number of $B$-coding blocks included in the range of $p$ in $x_j$ is denoted~$r^B_j(p)$
and we define $r^B(p) \coloneqq \sum_{j=1}^k r^B_j(p)$.
	
A position $p$ of $z$ is \emph{0-simple} (resp. \emph{1-simple}) if it is only matched to positions with value 0 (resp. 1) within $x_1,\ldots, x_k$, that is $\#_1(p)=0$ (resp. $\#_0(p)=0$). 
   It is \emph{simple} if it is 0- or 1-simple, and it is \emph{bad} otherwise. 
   The cost of a 1-simple position is~0 since it is matched only to positions with value 1.
   For a 0-simple position, we have $\#_1(p)=0$ and $k\leq \#_0(p)\leq 2k$ (more precisely, $\#_0(p) = k+r^B(p)$ and $r^B(p)\leq k$).
   Thus, the cost is between  $\frac k{k+1}$ and $\frac {2k}{2k+1}$.
   For $k>10$, the cost is always contained in the interval $[0.9, 1]$.

   We continue with several structural observations regarding a mean.
   
	\begin{observation} \label{lem:noConsecutiveSimple}
          There exists a mean without consecutive 1-simple positions or consecutive 0-simple positions. Such a mean
          is called \emph{non-redundant}.
	\end{observation}
	\begin{proof}
          Any two consecutive 1-simple (or 0-simple) positions in a mean $z$ have consecutive or intersecting ranges in each $x_j$ with the same value (1 or 0).
          Hence, they can be merged to one single 1-simple (or 0-simple) position.
	    
	    Since the warping of other positions in~$z$ remains unchanged, we focus on the cost of the merged position.
	    For 1-simple positions, the cost remains unchanged (both solutions yield a cost of 0 for the 1-simple positions). 
	    For 0-simple positions, the cost of the two 0-simple positions in the original solution is at least $0.9$ each. However, the cost of the merged 0-simple position is at most 1.
	\end{proof}

        We say a block $b$ of some input~$x_j$ is \emph{matched} (\emph{fully matched}) to a position~$p$ in~$z$ if some position (all positions) in $b$ is (are) matched to $p$.
        (Note that the distinction is only relevant for $B$-coding blocks, as all other blocks have size 1). 
	
	\begin{observation}	   
          For a non-redundant mean~$z$, any $B$-coding block of some~$x_j$ that is not fully matched to a position in~$z$ is matched to at least one bad position in~$z$.
          
           \label{two-consecutive-one-bad}
          If two consecutive positions in~$z$ are matched to a common block, then at least one of them is bad.
	\end{observation}
	\begin{proof}
	   Consider a $B$-coding block $b$ and all positions of $z$ matched to it. There are at least two of them, which cannot all be 0-simple (since~$z$ is non-redundant). Also, none of them can be 1-simple (since they are matched to at least one 0). Thus, at least one of them is bad.
	   
	   We prove the contrapositive of the second statement: if two simple positions $p < p'$ have a common block match, then they are both $a$-simple ($a\in\{0,1\}$), and cannot be consecutive in a non-redundant mean.
	   %For the second part of the statement, consider two simple positions $p < p'$ having a common match. Since they are both $a$-simple ($a\in\{0,1\}$), they cannot be consecutive (that is, $p' > p+1$).
       %    Note that also position~$p+1$ must have this match. Thus, it cannot be $(1-a)$-simple, and since it cannot be $a$-simple as well, $p+1$ is a bad position with $p<p+1<p'$.
	\end{proof}
	
	We now introduce an \emph{assignment} relation between a block $b$ of some $x_1,\ldots,x_k$ and a position~$p$ of~$z$. 	
	We say that $b$ is \emph{assigned} to the position $p$ (written $b\rightarrow p$) such that~$p$ is the leftmost simple position fully matched to~$b$ (if any), or (if no such simple position exists) such that $p$ is the leftmost bad position matched to $b$. Note that any size-1 block has at least one position (simple or bad) fully matched to it, and by \Cref{lem:noConsecutiveSimple}, any size-2 block has either a simple position fully matched to it or a bad position matched to it, so overall every block is assigned to exactly one position.

	For a position $p$ of~$z$, we introduce the following quantities:
\begin{align*}
%I'm replacing Box symbols by diamonds: I can't help interpreting the bos as a "missing character" in the font 
  \diamondsuit_0(p)&\coloneqq \text{number of 0-blocks in~$x_1,\ldots,x_k$ matched to $p$},\\
  \diamondsuit_1(p)&\coloneqq \text{number of 1-blocks in~$x_1,\ldots,x_k$ matched to $p$},\\
  q(p)&\coloneqq \text{number of $B$-coding blocks assigned to $p$},\\
  \rho_j(p)&\coloneqq\text{number of blocks in $x_j$ assigned to~$p$},\\
g(p)&\coloneqq 
\begin{cases}
0, &\text{ if $p$ is simple}\\
\max\{1, r_1(p) -1, \ldots, r_k(p)-1\}, &\text{ if $p$ is bad}
\end{cases}.
\end{align*}

We quickly observe the following:

\begin{enumerate}[(i)]
\item \label{rem1} If $p$ is simple, it is fully matched to $q(p)$ $B$-coding blocks.
	\\ \emph{Proof}. A $B$-coding block can only be fully matched to a single position, so if $p$ is simple, then all its fully matched $B$-coding blocks are assigned to it.

\item \label{rem2} If~$p$ is 0-simple, then $\diamondsuit_0(p)=k$, $\diamondsuit_1(p)=0$, and $C(p) = f_k(q(p))$.
	\\ \emph{Proof}. A 0-simple position is matched to exactly one 0-block in each~$x_1,\ldots,x_k$. The cost follows from (\ref{rem1}).
\item \label{rem3} If~$p$ is 1-simple, then $\diamondsuit_0(p)=0$, $\diamondsuit_1(p)=k$, and $C(p)=0$.
	\\ \emph{Proof}. A 1-simple position is matched to exactly one 1-block in each~$x_1,\ldots,x_k$. 
\item \label{rem5} If~$p$ is bad, then $g(p)\leq 2\min\{\diamondsuit_1(p), \diamondsuit_0(p)\}$.
\\ \emph{Proof}. For a bad position~$p$, we have $\min\{\diamondsuit_1(p),\diamondsuit_0(p)\}\geq 1$ on the one hand, and also $\min\{\diamondsuit_1(p),\diamondsuit_0(p)\}\geq \frac 12\max_{j=1,\ldots,k}(r_j(p)-1)$ on the other hand.
\item \label{rem9} If~$p$ is bad, then  $C(p)\geq \frac{1}{2}\min\{\diamondsuit_1(p)+1,\diamondsuit_0(p)\}$. 
\\ \emph{Proof}. Let $\mu = \min \{\diamondsuit_1(p)+1,\diamondsuit_0(p)\}$. Then $p$ is matched to at least $\mu$ 0's and $\mu$ 1's. Thus, $C(p)\geq \mu z[p]^2 + \mu(1-z[p])^2= \mu (z[p]^2 + (1-z[p])^2) \geq \frac \mu 2$.
\item \label{rem6} For every position~$p$, it holds  $|\diamondsuit_0(p)-\diamondsuit_1(p)| \leq k$ and $\diamondsuit_0(p)+\diamondsuit_1(p) \geq k$.
  \\ \emph{Proof}. For each~$x_j$, $1\le j \le k$, the difference between number of 0-blocks matched to~$p$ and number of 1-blocks matched to~$p$ is at most one. On the other hand, there is at least one 0- or 1-block matched to $p$ in $x_j$.
\end{enumerate}

   \paragraph{Cost of a Single Position.}
   We consider a fixed position $p$ of $z$.
   For simplification we write $\diamondsuit_0:=\diamondsuit_0(p), \diamondsuit_1:=\diamondsuit_1(p), q:=q(p), g:=g(p)$, and $C:=C(p)$.   
   The goal is to provide a lower bound for~$C$ that can be decomposed into the following elements: 
   \begin{itemize}
     \item a \emph{background cost} $\Cback(p) := \diamondsuit_0\frac{f_k(0)}k$, 
     \item a \emph{coding cost} $\Ccode(p) := \phi_k(q)$ reflecting the extra cost induced by a matched coding block, 
     \item a \emph{gap cost} of $\Cgap(p) := 0.01g$ which is added when $p$ is bad and increases as $p$ is matched to more blocks in some~$x_j$.
   \end{itemize}

   \begin{claim} \label{lem:LBsinglePosition}
     The cost $C$ of position $p$ is at least $LB(\diamondsuit_0,q,g,k)$, defined as follows:
     \begin{align*}     
 LB(\diamondsuit_0,q,g,k)&:=\Cback(p) + \Ccode(p) + \Cgap(p) \\
            &=\frac{\diamondsuit_0}{k+1} + \frac{q}{(k+q+1)(k+1)} + 0.01g  
     \end{align*}
   \end{claim}
   
   \begin{proof}
     In the following, we write~$LB$ for $LB(\diamondsuit_0,q,g,k)$.
     We prove $C\geq LB$ by case distinction.
   
   For a 0-simple position~$p$, we have $\diamondsuit_0 = k$ and $g=0$. Thus (by (\ref{rem2})),
   \begin{align*}
     LB= f_k(0) + \phi_k(q) + 0 =f_k(q)=C.
   \end{align*}      
   
   For a 1-simple position~$p$, we have~$\diamondsuit_0=0$, $q=0$ and $g=0$. Thus (by (\ref{rem3})),
   \begin{align*}
     LB= 0+\phi_k(0) + 0=0=C.
   \end{align*}    
   
   It remains to consider the case where $p$ is a bad position. Thus, we have $\diamondsuit_0\geq 1$ and  $\diamondsuit_1\geq 1$. First, note that
   \begin{align*}
    LB - 0.01g 
      =\frac{\diamondsuit_0}{k+1} +  \frac{q}{(k+q+1)(k+1)}
      \leq \frac{\diamondsuit_0}{k} +  \frac{1}{k+1}
      \leq 2\frac{\diamondsuit_0}{k},     
   \end{align*}
   that is, $LB \leq  2\frac{\diamondsuit_0}{k} + 0.01g$.
   
   We now use this upper bound for the following three sub-cases. First, if $\diamondsuit_0\leq \diamondsuit_1$, then we have   
   \begin{align*}   
		C&\geq \frac{1}{2}\diamondsuit_0 \text{\quad (by (\ref{rem9}))} \\
		 &\geq \frac {2}{k}\diamondsuit_0 + 0.02\diamondsuit_0 \text{\quad (since $k\geq 5$)} \\
		 &\geq LB \text{\quad (by (\ref{rem5}))} \\
   \end{align*}    
   
    Second, if $6\leq \diamondsuit_1 < \diamondsuit_0$, we have   
   \begin{align*}   
		C&\geq \frac{1}{2}\diamondsuit_1 \text{\quad (by (\ref{rem9}))}\\
		 &\geq 2+(\frac 2k +0.02)\diamondsuit_1 \text{\quad (since $k\geq 15$ and $\diamondsuit_1\geq 6$)} \\
		 &= \frac {2(\diamondsuit_1+2k)}k + 0.02\diamondsuit_1  \\
		 &\geq  \frac {2\diamondsuit_0}k+0.02\diamondsuit_1 \text{\quad(by (\ref{rem6}))}\\
		 &\geq LB \text{\quad (by (\ref{rem5}))}
   \end{align*}    
   
   Finally, if $\diamondsuit_1\leq 5 < \diamondsuit_0$, then $k-5\leq \diamondsuit_0\leq k+5$ (by (\ref{rem6})), and $g\leq 10$ (by (\ref{rem5})).
   We have
   \begin{align*}   
     C= \frac{(\diamondsuit_1+1)\#_0(p)}{\diamondsuit_1+1+\#_0(p)}\geq \frac{2(k-5)}{k-3}\ge 1.67
   \end{align*}
   using \Cref{lem:monotmonousCost}, with $\diamondsuit_1\geq 1$, $\#_0(p)\geq \diamondsuit_0\geq k-5$, and~$k \ge 15$.
   On the other side, we have
   \begin{align*}
     LB&= f_k(0)\frac{\diamondsuit_0}k + \phi_k(q) + 0.01 g \\
       &\leq \frac{5}{k}f_k(0) + f_k(0)+\phi_k(q) + 0.1 \text{\quad (since $\diamondsuit_0\leq k+5$ and $g\leq 10$)} \\
       &\leq \frac{5}{k} + 1.1  \text{\quad (since $f_k(0)\leq f_k(0)+\phi_k(q)\leq 1$)} \\
       &\leq 1.44 \text{\quad (since $k\geq 15$)} \\		
       &\leq C
   \end{align*}
   \end{proof}

                 %   We define the background, coding and gap costs of~$I$, denoted respectively $\Cback(I)$, $\Ccode(I)$ and $\Cgap(I)$, as the sum of the respective costs of its positions.

\paragraph{(Ir)regular Intervals.}
Two positions $p,p'$ of~$z$ at distance $\ell$ (i.e.\ $|p-p'|=\ell$) form an \emph{irregular pair} if they are matched respectively with blocks $y$ and $y'$ of some~$x_j$ having distance such that either $|y'-y| \leq \ell-\frac{m}{2k}$ or  $|y'-y|\geq \ell+\frac{m}{2k}$.
An interval~$I$ of positions in~$z$ is called \emph{regular} if it does not contain any irregular pair (otherwise it is called \emph{irregular}).
The background, coding and gap cost of~$I$ is the sum of the respective costs of its positions.
The total number of matches between blocks of~$x_1,\ldots,x_k$ and positions in~$I$ is denoted~$W(I)=\sum_{p\in I}(\diamondsuit_0(p)+\diamondsuit_1(p))$.

We aim at computing lower bounds on the cost of intervals.
The structure of regular intervals allows us to bound the coding cost using the minimum cost of the original \fMSCS[$\phi$] instance.
Irregular intervals have bad positions, which allow us to derive a lower bound on their gap cost.

We first introduce some notation: a position $p$ of $z$ is \emph{$j$-coding} if there is a coding block $y$ in $x_j$ such that $y$ is assigned to~$p$ (i.e., $y\rightarrow p$); it is \emph{coding} if it is $j$-coding for some $j$, otherwise it is \emph{non-coding}. A non-coding position~$p$  is \emph{free} if all positions $p'$ at distance at most $\frac{m}{2k}+2$ from $p$ are non-coding (see \Cref{fig:regularInterval}).
We first make the following technical claim before proving the main bound on the coding cost of regular intervals (\Cref{lem:regularCosts}).

\begin{claim}\label{count-coding-positions}
	In a regular interval, if two positions $p < p'$ are at distance at most $2im-\frac{m}{2k}$ for some $i$, then for any $j$, there are at most $i$ $j$-coding positions in $[p, p']$.
	
	Conversely,  if $p$ and $p'$  are at distance at least $2im+\frac{m}{2k}+1$, then, for any $j$, there are at least $i$ $j$-coding positions in $[p, p']$. 
\end{claim}

\begin{proof}
  Fix  $j\in\{1,\ldots,k\}$ and consider, in $x_j$, the first block $y$ matched to~$p$ and the last block $y'$ matched to $p'$ (then $y'\geq y$). Note that all $j$-coding positions in $[p,p']$ have been assigned to a distinct coding block in~$[y,y']$.
  Since $p$ and $p'$ are a regular pair, it holds
  \[y'-y< p'-p + \frac{m}{2k} \leq 2im.\]
  So $y'< y+2im$, and thus $x_j$ contains at most~$i$ coding blocks in~$[y,y']$. These coding blocks are assigned to positions in~$[p,p']$. Hence, $[p,p']$ contains at most $i$ $j$-coding positions.
	
  For the other direction, consider again blocks $y$ and $y'$ as above.
  In this case there is a slight difference: if block $y$ or $y'$ is coding, it might be assigned to a coding position outside of the interval $[p,p']$, which then would not count in the lower bound. Thus, we consider only blocks strictly between $y$ and $y'$, among which all coding blocks are assigned to a coding position in~$[p,p']$. 
  Since $p'-p\geq 2im+\frac{m}{2k}+1$, we have $y'-y>p'-p+\frac{m}{2k} \geq 2im+1$, so there are at least $2im$ blocks strictly between $y$ and $y'$, including at least $i$ coding blocks: they are assigned to at least $i$ $j$-coding positions in $[p,p']$.
\end{proof}

\begin{figure}
  \centering
	\begin{tikzpicture}
	
	\pgfmathsetseed{1} %1 is a good seed :)
	
	\foreach \j/\col in {1/brown,2/purple,3/green!50!black} {
		\node at (-1,3.5-\j) {$x_\j$};
		\draw (0,3.5-\j+.2) -- ++(11, 0);
		\draw (0,3.5-\j-.2) -- ++(11, 0);
		\foreach \i in {0,..., 5} {
			\node (code\j\i)[circle, fill=\col, inner sep=2pt] at (2*\i+.5,3.5-\j){};  	
		} 	   
		\foreach \i in {0,..., 4} {
			\node (pos\j\i)[circle, fill=\col, inner sep=2pt] at (2*\i+.02+.25*rand+.7*\j,-1){};
		} 	   
	}
	
	\draw [decorate,decoration={brace,amplitude=10pt},yshift=0pt]
	(0.6,3-4+.25) -- ++(10.1,0) node [black,midway,yshift=.5cm] {\footnotesize
		$I$};
	
	\fill[cyan] (0,3-4+.2) rectangle ++(.6,-.4);
	\fill[cyan] (11,3-4+.2) rectangle ++(-.3,-.4);
	\fill[pattern=north west lines, pattern color=orange] (1.5,3-4+.2) rectangle ++(.4,-.4);
	
	\node at (-1,3-4) {$z$};
	\draw (0,3-4+.2) -- ++(11, 0);
	\draw (0,3-4-.2) -- ++(11, 0);
	
	\foreach \i in {0,..., 4} {
		\draw (2*\i+1.7,-.8) -- ++ (0, -.4);
		\node [] at (2*\i+1.7,-1.4){$p_\i$};
		%\node [circle, fill=brown, inner sep=2pt] at (2*\i+.5,3-4){};
	} 	   
	\draw[->, densely dotted] (code11)  to [bend left=15] (pos11);
	\draw[->, densely dotted] (code22)  to [bend right=15] (pos21);
	\draw[->, densely dotted] (code30)  to [bend left=15] (pos30);
	
	\draw[->, densely dotted] (code14)  to [bend left=15] (pos14);
	\draw[->, densely dotted] (code25)  to [bend right=15] (pos24);
	\draw[->, densely dotted] (code33)  to [bend left=15] (pos33);
	
	\end{tikzpicture}	
	\caption {\label{fig:regularInterval} Decomposition of a regular interval ($I$) into segments delimited at positions $p_i$. Coding blocks in $x_j$ are indicated with colored bullets, as well as their assigned positions in~$z$ (dotted arcs). Position $p_0$ is free (no coding position within the striped area), and each segment $[p_i, p_{i+1}]$ contains exactly one $j$-coding position from each $x_j$. These coding positions correspond to columns of a multiple circular shift of the \fMSCS[$\phi$] input strings. The coding cost of $n$ consecutive segments can be lower-bounded by the minimum cost of the \fMSCS[$\phi$] instance. 
	}
\end{figure}
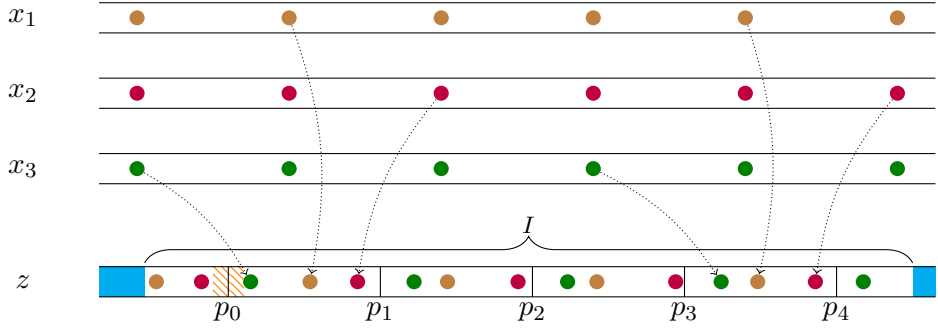

\begin{claim}\label{lem:regularCosts}
  The coding cost of a length-$\ell$ regular interval is at least
  \[\left(\frac{(\ell+1)}{2mn} - 2\right)(c+\epsilon).\]
\end{claim}
\begin{proof}
Let $I$ be a length-$\ell$ interval. Assume that $\ell\geq 3mn$ (otherwise the stated lower bound is negative which is trivial).

The first part of the proof consists of splitting interval~$I$ into consecutive length-$2m$ segments, each one containing exactly one $j$-coding position for each $j$. To this end, a few positions need to be cropped from both ends of~$I$. In other words, we need to find a good starting point (a free position) close to the left end of~$I$.

Consider the first position $p$ of $I$ and position $p':=p+2m-\frac{m}{2k}$ (in~$I$).  
By \Cref{count-coding-positions}, for any $j$, there is at most one $j$-coding position $p''\in[p,p']$ (and so at most $k$ such coding positions in total). 
Accordingly, the interval $[p,p']$ contains at most $(k+1)$ disjoint intervals of non-coding positions, with total size at least $(p'-p +1) - k =  2m-\frac{m}{2k}-k+1$. 
Using the fact that $m$ is, by definition, much larger than $k$ (specifically, $\frac m4\geq 5(k+1)+\frac 1k \Rightarrow 2mk-\frac m2 -1 \geq m(k+1)+5k(k+1)$ when $k\geq 2$), there is one interval of non-coding positions in $[p,p']$ with size at least
\[\frac{2m-\frac{m}{2k} -k+1}{k+1} 
  \geq\frac {2mk -\frac m2-1}{k(k+1)} 
  \geq \frac {m(k+1)+5k(k+1)}{k(k+1)}
  \geq \frac m{k} +5 
  ,\]
so its middle position is free.
Hence, there exists a free position, denoted $p_0$, with $p+\frac m{2k}+2 \leq p_0 \leq p+2m-\frac m{2k}-2$. 

Let $\lambda:=\lfloor \frac{\ell}{2m}\rfloor -2$ and $p_i:=p_0+2im$ for $0<i\leq \lambda$. Note that $\lambda\geq n$ (following the assumption on $\ell$), and that every $p_i$ is in $I$.

We fix some input sequence $x_j$ for $1\leq j\leq k$.
Intuitively, positions $p_i$ are the cutting points of our segments within interval $I$.
We now aim at showing that there is exactly one $j$-coding position in each segment~$[p_{i-1},p_i-1]$.
First, consider positions $p=p_0-\frac{m}{2k}-2$  and $p_i-1$. 
By \Cref{count-coding-positions}, since $(p_i-1)-p= 2im + \frac{m}{2k}+1$, interval $[p,p_i-1]$ contains at least $i$ $j$-coding positions.
Since $p_0$ is free, these coding positions cannot be before $p_0$ (as $p_0-p= \frac{m}{2k}+2$), so they are in $[p_0, p_i-1]$.
Consider now positions $p'=p_0+\frac{m}{2k}+1$  and $p_i-1$.
By \Cref{count-coding-positions}, since $(p_i-1)-p= 2im-\frac{m}{2k}$, there are at most $i$ $j$-coding positions in $[p',p_i-1]$, and therefore at most $i$ $j$-coding positions in $[p_0, p_i-1]$. Overall there are exactly $i$ $j$-coding positions in $[p_0, p_i-1]$.

Overall, there is exactly one $j$-coding position in $[p_{i-1}, p_{i}-1]$ for every $0<i\leq \lambda$ and every $j$.
We write~$C_{i,j}$ for the corresponding coding block in~$x_j$.

Let $q_i$ be the number of $B$-coding blocks among $C_{i,1},\ldots,C_{i,k}$.
Then,
\[\sum_{p=p_{i-1}}^{p_{i}-1}q(p) = q_i \text{ and } \sum_{p=p_{i-1}}^{p_{i}-1}\phi_k(q(p))\geq \phi_k(q_i).\] 

Let $\delta_j$ be such that $C_{0,j}$ is the $\delta_j$-th coding block of $x_j$.
Then, $C_{i,j}$ is the $(\delta_j+i)$-th coding block of $x_j$, for every $0\leq i<\lambda$, and $(C_{i,1},\ldots,C_{i,k})$ corresponds to the $(i\bmod n)$-th column in the multiple circular shift $\Delta=(\delta_1,\ldots,\delta_k)$ of $s_1,\ldots, s_k$.
Thus, $q_i$ is the number of $B$s in this column and $\phi_k(q_i)$ is the corresponding cost of this column. Summing over $n$ consecutive columns (circularly), we get, for any integer $a$ with $0\leq a\leq \lambda-n$,  a multiple circular shift of $s_1, \ldots, s_k$ with cost $\sum_{i=a}^{a+n-1}\phi_k(q_i)$.
Since, by assumption, every multiple circular shift of~$s_1, \ldots, s_k$ has cost at least $c+\epsilon$, we have
\[\sum_{i=a}^{a+n-1}\phi_k(q_i)\geq c+\epsilon.\]
We can now compute the lower bound on the coding cost of interval $I$. To this end, we first extract $\left\lfloor \frac \lambda n\right \rfloor\ge 1$ length-$2mn$ subintervals of~$I$, each consisting of $n$ segments of the form $[p_{i-1},p_{i}-1]$. It follows

\begin{align*}
 \sum_{p\in I} \phi_k(q(p))
 &\geq \sum_{p=p_0}^{p_{\lambda}-1} \phi_k(q(p))\\
 & \geq\sum_{a=0}^{\left\lfloor \frac \lambda n\right\rfloor-1}  
       \sum_{i=an}^{an+n-1}
       \sum_{p=p_{i-1}}^{p_{i}-1} \phi_k(q(p))\\
 & \geq\sum_{a=0}^{\left\lfloor \frac \lambda n\right\rfloor-1}  
	\sum_{i=an}^{an+n-1}
	 \phi_k(q_{i})\\
 & \geq\sum_{a=0}^{\left\lfloor \frac \lambda n\right\rfloor-1}  
	c+\epsilon \\
 & =   \left\lfloor \frac \lambda n\right\rfloor(c+\epsilon)\\
 &\geq  \left(\frac \lambda n-1\right)(c+\epsilon) \\ 
 &=  \left(\frac {\lfloor \frac{\ell}{2m}\rfloor -2} n-1\right)(c+\epsilon) \\
 &\geq \left(\frac { \frac{\ell}{2m} -3} n-1\right)(c+\epsilon) \\
 &\geq \left(\frac \ell{2mn} - 2\right)(c+\epsilon)
\end{align*}

\end{proof}

Concerning the gap costs, we prove the following.

\begin{claim}\label{lem:IrregularGap}
For any interval $I$ of length $\ell$, the gap cost is at least	
	
\[\Cgap(I) \geq \frac1{100}\left(\frac {W(I)}k-\ell\right).\]

Moreover, if $I$ is irregular, then
\[\Cgap(I)\geq\frac m{400k}.\]
\end{claim}
\begin{proof}
For the first lower bound, it suffices to note that for any position~$p$ (simple or bad), it holds $g(p)\geq \frac{W(p)} k-1$, where $W(p)$ is the overall number of blocks matched to~$p$.

For the second lower bound, consider an irregular pair $p, p'$ in~$I$ and an integer $j$, such that $p$ is matched to some block $y$ in $x_j$ and $p'$ is matched to block $y'$ in~$x_j$ with $|(y'-y)-(p'-p)|>\frac m{2k}$. 

If  $y'-y>p'-p + \frac m{2k}$, then
\[\sum_{p''=p}^{p'} g(p'') \geq \sum_{p''=p}^{p'} (r_j(p'')-1)\geq y'-y - (p'-p)>\frac m{2k}>\frac m{4k}.\]

If  $y'-y<p'-p - \frac m{2k}$, then there are at least $\frac m{2k}$ pairs of consecutive positions matching the same block in $x_j$, and for every such pair at least one of the two positions is bad (by \Cref{two-consecutive-one-bad}). Since any bad position may be counted in at most two such pairs, the interval has at least $\frac m{4k}$ bad positions. Hence, using $g(p)\geq 1$ for bad positions, we obtain $\sum_{p''=p}^{p'} g(p'') \geq \frac m{4k}$.

\end{proof}

\paragraph{Cost of a Mean.}
Now, consider the partition
\[[\alpha_1:=1,\beta_1],[\alpha_2:=\beta_1+1,\beta_2],\ldots,[\alpha_L:=\beta_{L-1}+1,\beta_L:=|z|]\] of~$[1, |z|]$ into smallest irregular intervals (except for $[\alpha_L,\beta_L]$ which is possibly regular if the last remaining subinterval is regular).
The following lower bounds hold.

\begin{claim}\label{lem:OneIntervalCost}
  For $1\le i<L$, the coding cost plus the gap cost of interval $[\alpha_i,\beta_i]$ is at least 
  \[\Cgap([\alpha_i, \beta_i])+\Ccode([\alpha_i, \beta_i])\geq (c+\epsilon) \frac{W([\alpha_i,\beta_i])}{2knm}.\]
  
  The coding and gap costs of~$[\alpha_L,\beta_L]$ sum to at least
  \[\Cgap([\alpha_L, \beta_L])+\Ccode([\alpha_L, \beta_L])\geq (c+\epsilon) \frac{W([\alpha_i,\beta_i])}{2knm} - 2(c+\epsilon).\]
\end{claim}

\begin{proof}
  Consider interval $[\alpha_i,\beta_i]$, let $\ell$ be its length, and $W:=W([\alpha_i,\beta_i])$.
  Since $[\alpha_i,\beta_i-1]$ is a regular interval of length~$\ell-1$, by \Cref{lem:regularCosts}, we have the following lower bound on the coding cost: 
\[
 \Ccode([\alpha_i, \beta_i]) \geq  \Ccode([\alpha_i, \beta_i-1]) \geq \frac {\ell}{2mn} (c+\epsilon) -2(c+\epsilon).
\]

 For irregular intervals $[\alpha_i, \beta_i]$ with $1\le i<L$, we combine both bounds on the gap cost of \Cref{lem:IrregularGap} (by averaging their values):

\[\Cgap([\alpha_i, \beta_i]) 
 \geq \frac 1{200}\left(\frac{W}{k}-\ell\right)+ \frac m{800k}.\]
 
 Using $m\geq 1600k(c+\epsilon)\geq \frac {200(c+\epsilon)}{2n}$, we obtain 
 
 \[\Cgap([\alpha_i, \beta_i]) 
 \geq\frac{c+\epsilon}{2nm}\left(\frac{W}k -\ell\right) + 2(c+\epsilon).\]
 
 Summing with $\Ccode$ gives the desired lower bound for irregular intervals.  
 For interval $[\alpha_L, \beta_L]$, we use the general lower bound from \Cref{lem:IrregularGap}, which yields
 \[\Cgap([\alpha_L, \beta_L]) \geq \frac1{100}\left(\frac Wk-\ell\right) \geq \frac {c+\epsilon}{2mn}\left(\frac W{k}-\ell\right).\]
 Again, the sum of $\Cgap$ and $\Ccode$ above yields the desired lower bound for $[\alpha_L,\beta_L]$.
\end{proof}

Finally, to finish the proof, we show that a mean has high cost.

\begin{claim}
Let~$z$ be a mean and consider the interval $I=[1,|z|]$. Then, it holds
\begin{align*}
  \Ccode(I)+\Cgap(I) &\geq (c+\epsilon)r - 2(c+\epsilon) \text{ and}\\
  \Cback(I) &\geq nmrf_k(0).
\end{align*}
Hence, the total cost of~$z$ is strictly larger than $r(nmf_k(0) + c) + 3mnk=c'$.
\end{claim}
\begin{proof}
For the coding and gap cost, we use \Cref{lem:OneIntervalCost} together with the fact that all~$2knmr$ blocks of $x_1,\ldots, x_k$ are involved in at least one match with a position of~$z$, which yields $W(I) = \sum_{i=1}^L W([\alpha_i,\beta_i])\geq 2knmr$. 

The overall background cost is
\[\Cback(I) = \sum_{p=1}^{|z|} \diamondsuit_0(p)\frac{f_k(0)}k.\]
Since overall there are $knmr$ 0-blocks in $x_1,\ldots, x_k$, and each of those is matched to at least one position $p$ in $I$, we have $\sum_{p=1}^{|z|} \diamondsuit_0(p)\geq knmr$ and  $\Cback(I)  \geq nmrf_k(0)$.

The sum of these three costs gives a lower bound on the total cost of $z$ (using \Cref{lem:LBsinglePosition} on each of its positions).
This sum amounts to \[nmrf_k(0) + (c+\epsilon)r - 2(c+\epsilon) = nmrf_k(0) + cr + \epsilon r - 2(c+\epsilon).\] Since $\epsilon r > 3mnk + 2 (c+\epsilon)$, we get the desired bound.
\end{proof}

Since the above reduction is a polynomial-time reduction from~\fMSCS[$\phi$] where the resulting number of time series is linearly bounded in the number of strings in the \fMSCS[$\phi$] instance,
\Cref{thm:DTW-hard} now follows from \Cref{cor:MSCShardness}.
\end{proof}

Closing this section, we remark that Buchin et al.~\cite[Theorem~7]{BDS19} recently obtained the same hardness results for the problem of computing an average series~$z$ that minimizes
\begin{align*}
 \Fcost_p^q(z)\coloneqq\sum_{j=1}^{k}\left(\min_{p_j\in\mathcal{P}_{|x_j|,|z|}}\sum_{(u,v)\in p_j}|x_j[u]-z[v]|^p\right)^{q/p}
\end{align*}
for all integers $p,q \ge 1$. Their reduction, however, builds time series containing three different values.
Hence, \Cref{thm:DTW-hard} yields a stronger hardness on binary inputs for $p=q=2$.
Note that if also the mean is restricted to be a binary time series, then the problem is solvable in polynomial time~\cite[Theorem~1]{BFFJNS19}.

\section{Conclusion}\label{sec:concl}
Shedding light on the computational complexity of prominent consensus problems 
in stringology and time series analysis,  
we proved several tight computational hardness results for circular string 
alignment problems and time series averaging in dynamic time warping spaces.
Notably, we have shown that the complexity of consensus string problems can drastically change (that is, they become hard) when considering \emph{circular} strings 
and \emph{shift} operations instead of classic strings.
Our results imply that these problems with a rich set of applications are intractable in the worst case (even on binary data). Hence, it is unlikely to find algorithms which significantly improve the worst-case running time of the best known algorithms so far. 
This now partly justifies the use of heuristics as have been used for a long time
in many real-world applications.

We conclude with some open questions and directions for future work.
\begin{itemize}
\item We conjecture that the idea of the reduction for \fMSCS can be used to prove the same hardness result for most non-linear (polynomially bounded) order-independent cost functions (note that \fMSCS is trivially solvable if~$f_k$ is linear since every shift has the same cost).
  Proving a complexity dichotomy is an interesting goal to achieve.
  \item From an algorithmic point of view, it would be nice to improve the constants in the exponents of the running times, that is, to find algorithms running in time~$O(n^{\alpha k})$ for small~$\alpha$. In particular, for \DTW, we ask to find an $O(n^k)$-time algorithm.
    \item What about the parameter maximum sequence length~$n$? Are the considered
 problems polynomial-time solvable if $n$~is a constant, or are they even 
 fixed-parameter tractable with respect to~$n$?
 \item Finally, can the hardness result for averaging time series with respect to~$(p,q)$-DTW by Buchin et al.~\cite[Theorem~7]{BDS19} be strengthened to binary inputs?
\end{itemize}

\bibliographystyle{abbrvnat}
\bibliography{ref}

\end{document}